\documentclass{llncs}
\usepackage{amssymb, amsmath,graphicx,graphics,caption,enumerate}
\usepackage{comment}
\usepackage[utf8]{inputenc} 
\usepackage[T1]{fontenc}
\usepackage{tikz}
\usepackage{hyperref}
\usepackage{tkz-graph}
\usepackage{array}
\usepackage{adjustbox}

\usepackage{times}
\usepackage{microtype}

\newcommand{\NP}{{\sf NP}}
\DeclareMathOperator{\tw}{tw}
\DeclareMathOperator{\cw}{cw}
\newcommand{\ssi}{\subseteq_i}
\newcommand{\si}{\supseteq_i}

\newcommand\displaycase[1]{{\bf #1}}

\newcolumntype{V}{>{\centering\arraybackslash} m{.17\textwidth} }
\newcolumntype{W}{>{\centering\arraybackslash} m{.13\textwidth} }

\newenvironment{enumeratei}
    {\begin{enumerate}[(i)]

    }
    {\end{enumerate}}

\newenvironment{enumerate1}
    {\begin{enumerate}[1.]

    }
    {\end{enumerate}}

\renewcommand\subparagraph[1]{{\bf #1. }}
\newtheorem{oproblem}{Open Problem}

\newcounter{ctrclaim}[theorem]

\title{Clique-width of Graph Classes Defined by\\ Two Forbidden Induced Subgraphs\thanks{The research in this paper was supported by EPSRC (EP/G043434/1 and EP/K025090/1) and ANR (TODO ANR-09-EMER-010).}}

\author{Konrad K. Dabrowski\inst{1} \and Dani\"el Paulusma\inst{1}}
\institute{
School of Engineering and  Computing Sciences, Durham University,\\
Science Laboratories, South Road,\\
Durham DH1 3LE, United Kingdom\\
\texttt{\{konrad.dabrowski,daniel.paulusma\}@durham.ac.uk}
}
\begin{document}
\maketitle
\begin{abstract}
\begin{sloppypar}
If a graph has no induced subgraph isomorphic to any graph in a finite family $\{H_1,\ldots,H_p\}$, it is said to be $(H_1,\ldots,H_p)$-free.
The class of $H$-free graphs has bounded clique-width if and only if $H$ is an induced subgraph of the 4-vertex path $P_4$. 
We study the (un)boundedness of the clique-width of graph classes defined by two forbidden induced subgraphs~$H_1$ and $H_2$.
Prior to our study it was not known whether the number of open cases was finite.
We provide a positive answer to this question.
To reduce the number of open cases we determine new graph classes of bounded clique-width 
and new graph classes of unbounded clique-width.
For obtaining the latter results we first present a new, generic construction for graph classes of unbounded clique-width.
Our results settle the boundedness or unboundedness of the clique-width of the class of $(H_1,H_2)$-free graphs
\end{sloppypar}
\begin{enumeratei}
\begin{sloppypar}
\item for all pairs $(H_1,H_2)$, both of which are connected, except two non-equivalent cases, and
\item for all pairs $(H_1,H_2)$, at least one of which is not connected, except~11 non-equivalent cases.
\end{sloppypar}
\end{enumeratei}
We also consider classes characterized by forbidding a finite family of graphs $\{H_1,\ldots,H_p\}$ as subgraphs, minors and topological minors, respectively, and completely determine which of these classes have bounded clique-width.
Finally, we show algorithmic consequences of our results 
for the graph colouring problem restricted to $(H_1,H_2)$-free graphs.
\end{abstract}

\section{Introduction}\label{s-intro}

Clique-width is a well-known graph parameter studied both in a structural and in an algorithmic context; we refer to the surveys of Gurski~\cite{Gu07} and Kami\'nski, Lozin and Milani\v{c}~\cite{KLM09} for an in-depth study of the properties of clique-width. 
However, our understanding of clique-width, which is one of the
most difficult graph parameters to deal with, is still very limited. For example, no
polynomial-time algorithms are known for computing the clique-width of very
restricted graph classes, such as unit interval graphs, or for deciding whether
a graph has clique-width at most $c$ for any fixed $c\geq 4$ (as an aside, we
note that such an algorithm does exist for $c=3$~\cite{CHLRR12}).

In order to get more structural insight into clique-width, we are interested in determining 
whether the clique-width of some given class of graphs is {\it bounded}, that is, whether there exists
a constant~$c$ such that every graph from the class has clique-width at most~$c$
(our secondary motivation is algorithmic, as we will explain in detail later).
The graph classes that we consider consist 
of graphs in which one or more specified graphs are forbidden as a ``pattern''.
In particular, we consider classes of graphs that contain no graph from some specified family $\{H_1,\ldots,H_p\}$ as an {\it induced subgraph}; such classes are said to
be $(H_1,\ldots,H_p)$-free. Our research is well embedded in the literature, as there are many papers that determine the boundedness or unboundedness of the clique-width of graph classes characterized by one or more forbidden induced subgraphs; 
see e.g.~\cite{BL02,BGMS14,BDHP15,BELL06,BKM06,BK05,BLM04b,BLM04,BM02,BM03,DGP14,DHP0,DLRR12,DP14,GR99b,LR04,LR06,LV08,MR99}.

As we show later, it is not difficult to verify that the class of $H$-free graphs has bounded clique-width if and only if $H$ is an induced subgraph of the 4-vertex path $P_4$. Hence, it is natural to consider the following problem:

\smallskip
\noindent
{\em For which pairs $(H_1,H_2)$ does the class of $(H_1,H_2)$-free graphs have bounded clique-width?}

\medskip
\noindent
In this paper we address this question by narrowing the gap between the known and open cases significantly; 
in particular we show that the number of open cases is finite.
We emphasise that the {\it underlying} research question is: what kind of properties of a graph class ensure that its clique-width is bounded?
Our paper is to be interpreted as a further step towards this direction, and
in our research project (see also~\cite{BDHP15,DHP0,DP14}) we aim 
to develop 
{\it general} techniques for attacking a number of the open cases simultaneously.

\medskip
\noindent
\subparagraph{Algorithmic Motivation}
For problems that are \NP-complete in general, one naturally seeks to find subclasses of graphs on which they are tractable,
and graph classes of bounded clique-width have been studied extensively for this purpose, as we discuss below.

Courcelle, Makowsky and Rotics~\cite{CMR00} showed that all MSO$_1$ graph problems, which are
problems  definable in Monadic Second Order Logic using quantifiers on vertices but not on edges, 
can be solved in linear time on graphs with clique-width at most~$c$, provided that a $c$-expression of the input graph is given. Later, Espelage, Gurski and Wanke~\cite{EGW01}, Kobler and Rotics~\cite{KR03b} and Rao~\cite{Ra07} proved the same result for many non-MSO$_1$ graph problems.
Although computing the clique-width of a given graph is 
\NP-hard, as shown by Fellows, Rosamond, Rotics and Szeider~\cite{FRRS09}, it is possible to find an $(8^{c}-1)$-expression for any $n$-vertex graph with clique-width at most $c$ in cubic time. This is a result of Oum~\cite{Oum08} after a similar result (with a worse bound and running time) had already been shown 
by~Oum and Seymour~\cite{OS06}.
Hence, the \NP-complete problems considered in the aforementioned papers~\cite{CMR00,EGW01,KR03b,Ra07} are all polynomial-time solvable on any graph class of bounded clique-width even if no $c$-expression of the input graph is given.

As a consequence of the above, when solving an \NP-complete problem on some  graph class~${\cal G}$, it is natural to try to determine {\em first} whether the clique-width of ${\cal G}$ is bounded. 
In particular this is the case if we aim to determine the computational complexity of some \NP-complete problem when restricted to graph classes characterized by some common type of property.
This property may be the absence of a family of forbidden induced subgraphs $H_1,\ldots,H_p$ and
we may want to classify for which families of graphs $H_1,\ldots,H_p$ the problem is still \NP-hard and for which ones it becomes polynomial-time solvable (in order
to increase our understanding of the hardness of the problem in general). 
We give examples later.

\medskip
\noindent
\subparagraph{Our Results}
In Section~\ref{sec:prelim} we state a number of basic results on clique-width and two results on $H$-free bipartite graphs that we showed in a very recent paper~\cite{DP14}; we need these results for proving our new results.
We then identify a number of new classes of $(H_1,H_2)$-free graphs of  bounded clique-width (Section~\ref{sec:bounded}) and
unbounded clique-width (Section~\ref{sec:unbounded}). 
In particular, the new unbounded cases are obtained from a new, general construction for graph classes of unbounded clique-width.
In Section~\ref{sec:twographs}, we first observe for which graphs $H_1$ the class of $H_1$-free graphs have bounded clique-width. We then 
present our main theorem that 
gives a summary of our current knowledge of those pairs $(H_1,H_2)$ for which the class of $(H_1,H_2)$-free graphs has  bounded
clique-width and unbounded clique-width, respectively.\footnote{Before finding the combinatorial proof of our main theorem we first obtained a computer-assisted proof using Sage~\cite{sage} and the Information System on Graph Classes and their Inclusions~\cite{isgci} (which keeps a record of classes for which boundedness or unboundedness of clique-width is known). In particular, we would like to thank Nathann Cohen and 
Ernst de Ridder for their help.}
In this way we are able to narrow the gap to~13 open cases (up to some equivalence relation, which we explain later); when we only consider
pairs $(H_1,H_2)$ of connected graphs the number of non-equivalent open cases is only two.
In order to present our summary, we will need several results from the papers listed above.
We will also need these results in Section~\ref{s-other}, where we consider graph classes characterized by forbidding a finite family of graphs $\{H_1,\ldots,H_p\}$ as subgraphs, minors and topological minors, respectively. For these containment relations we are able to completely determine which of these classes have bounded clique-width.

\medskip
\noindent
\subparagraph{Algorithmic Consequences} Our results are of interest for any \NP-complete problem that is solvable in polynomial time on graph classes of bounded clique-width.
In Section~\ref{s-conclusions} we give a concrete application of our results by considering the well-known {\sc Colouring} problem, which 
is that of testing whether a graph can be coloured with at most $k$ colours for some given integer $k$ and which is 
solvable in polynomial time on any graph class of bounded clique-width~\cite{KR03b}.
The complexity of {\sc Colouring} has been studied extensively for $(H_1,H_2)$-free  graphs~\cite{DGP14,DLRR12,GP14,KKTW01,Ma13,Sc05}, but
a full classification is still far from being settled.
Many of the polynomial-time results follow directly from bounding the clique-width in such classes.
As such this forms a direct motivation for our research.
Another example for which our study might be of interest is the {\sc List $k$-Colouring} problem 
(another problem mentioned in the paper of Kobler and Rotics~\cite{KR03b}). The complexity of this problem was recently investigated for $(H_1,H_2)$-free graphs when $H_1$ is a path and $H_2$ is a cycle~\cite{HJP14}.

\medskip
\noindent
\subparagraph{Related Work}
We finish this section by briefly discussing some related results. 

First,
a graph class ${\cal G}$ has power-bounded clique-width if there is a constant $r$ so that the class consisting of all $r$-th powers of all graphs from ${\cal G}$ has bounded clique-width. Recently, Bonomo, Grippo, Milani\v{c} and Safe~\cite{BGMS14} determined all pairs  
of connected graphs $H_1,H_2$ 
for which the class of $(H_1,H_2)$-free graphs has power-bounded clique-width. If a graph class has bounded clique-width, it has power-bounded clique-width.
However, the reverse implication does not hold in general. The latter can be seen as follows. Bonomo et al.~\cite{BGMS14} showed that the class of $H$-free graphs has power-bounded clique-width if and only if $H$ is a linear forest (recall that such a class has bounded clique-width if and only if $H$ is an induced subgraph of $P_4$). Their classification for connected graphs~$H_1,H_2$ is the following. Let~$S_{1,i,j}$ be the graph obtained from a 4-vertex star by subdividing one leg $i-1$ times and another leg $j-1$ times. Let $T_{1,i,j}$ be the line graph of~$S_{1,i,j}$. Then the class of $(H_1,H_2)$-free graphs has power-bounded clique-width if and only if one of the  following two cases applies: (i) one of $H_1, H_2$ is a path or (ii) one of~$H_1, H_2$ is isomorphic to
$S_{1,i,j}$ for some $i,j\geq 1$ and the other one is isomorphic to $T_{1,i',j'}$ for some $i',j'\geq 1$. 
In particular, the classes of power-unbounded clique-width were already known to have unbounded clique-width.

Second, Kratsch and Schweitzer~\cite{KS12} initiated a study into the computational complexity of the {\sc Graph Isomorphism} problem (GI) for graph classes defined by two forbidden induced subgraphs.
The exact number of open cases  is still not known, but
Schweitzer~\cite{Sc15} very recently proved that this number is finite.
There are similarities between classifying the boundedness of clique-width and solving GI for classes of graphs characterized by one or more forbidden induced subgraphs.
This was noted by Schweitzer\cite{Sc15}, who proved
that any graph class that allows a so-called simple path encoding has unbounded clique-width.
Indeed, a common technique (see e.g.~\cite{KLM09}) for showing that a class of graphs has unbounded clique-width relies on showing that it contains simple path encodings of walls or of graphs in some other specific graph class known to have unbounded clique-width.
For $H$-free graphs, GI is polynomial-time solvable  if $H$ is an induced subgraph of 
$P_4$~\cite{CLS81} and GI-complete otherwise~\cite{KS12}. Hence, if only one induced subgraph is forbidden, the dichotomy classifications for clique-width and GI are identical.

\section{Preliminaries}\label{sec:prelim}

Below we define the graph terminology used throughout our paper. For any undefined terminology we refer to Diestel~\cite{Di12}.

Let $G$ be a graph.  
The set $N(u)=\{v\in V(G)\; |\; uv\in E(G)\}$ is the {\em (open) neighbourhood} of $u\in V(G)$ and $N[u] = N(u) \cup \{u\}$ is the {\em closed
neighbourhood} of $u\in V(G)$.
The {\em degree} of a vertex in a graph is the size of its neighbourhood.
The {\em maximum degree} of a graph is the maximum vertex degree.
For a subset $S\subseteq V(G)$, we let $G[S]$ denote the subgraph of $G$ {\it induced} by $S$, which has vertex set~$S$ and edge set $\{uv\; |\; u,v\in S, uv\in E(G)\}$.
If $S=\{s_1,\ldots,s_r\}$ then, to simplify notation, we may also write $G[s_1,\ldots,s_r]$ instead of $G[\{s_1,\ldots,s_r\}]$.
Let~$H$ be another graph.
We write $H\ssi G$ to indicate that $H$ is an induced subgraph of~$G$.

Let  $\{H_1,\ldots,H_p\}$ be a set of graphs. We say that
a graph $G$ is {\it $(H_1,\ldots,H_p)$-free} if $G$ has no induced subgraph isomorphic to a graph in $\{H_1,\ldots,H_p\}$.
If $p=1$, we may write $H_1$-free instead of $(H_1)$-free.
The {\it disjoint union} $G+H$ of two vertex-disjoint graphs $G$ and $H$ is the graph with vertex set 
$V(G)\cup V(H)$ and edge set  $E(G)\cup E(H)$. 
We denote the disjoint union of $r$ copies of $G$ by $rG$.

For positive integers $s$ and $t$, the  {\em Ramsey number} $R(s,t)$ is the smallest number $n$ such that all  graphs on $n$
vertices contain an independent set of size~$s$ or a clique of size $t$.  Ramsey's Theorem~\cite{Ra30} states that such a number exists for all positive integers $s$ and $t$.

The {\em clique-width} of a graph $G$, denoted $\cw(G)$, is the minimum 
number of labels needed to
construct $G$ by
using the following four operations:
\begin{enumerate}
\item creating a new graph consisting of a single vertex $v$ with label $i$ (denoted by $i(v)$);
\item taking the disjoint union of two labelled graphs $G_1$ and $G_2$ (denoted by $G_1\oplus G_2$);
\item joining each vertex with label $i$ to each vertex with label $j$ ($i\neq j$, denoted by $\eta_{i,j}$);
\item renaming label $i$ to $j$ (denoted by $\rho_{i\rightarrow j}$).
\end{enumerate}
An algebraic term that represents such a construction of $G$ and uses at most~$k$ labels is said to be a {\em $k$-expression} of $G$ (i.e. the clique-width of $G$ is the minimum $k$ for which $G$ has a $k$-expression).
For instance, an induced path on four consecutive vertices $a, b, c, d$ has clique-width equal to 3, and the following 3-expression can be used to construct it:
$$
\eta_{3,2}(3(d)\oplus \rho_{3\rightarrow 2}(\rho_{2\rightarrow 1}(\eta_{3,2}(3(c)\oplus \eta_{2,1}(2(b)\oplus 1(a)))))).
$$
Alternatively, any $k$-expression for a graph $G$ can be represented by a rooted tree, where the leaves correspond to
the operations of vertex creation and the internal nodes correspond to the other three operations.
The rooted tree representing the above $k$-expression is depicted in \figurename~\ref{fig:cwtree}.
A class of graphs~${\cal G}$ has \emph{bounded} clique-width if 
there is a constant $c$ such that the clique-width of every graph in~${\cal G}$ is at most $c$; otherwise the clique-width of ${\cal G}$ is 
{\em unbounded}.

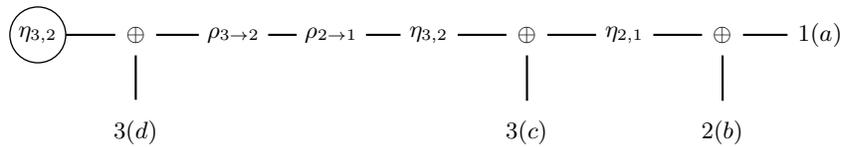
\begin{figure}
\begin{center}
\begin{tikzpicture}[scale=1.3]%[scale=0.4, every node/.style={scale=0.3}]
\GraphInit[vstyle=Empty]
\Vertex[x=5,y=1,L=$\oplus$           ]{v51}
\Vertex[x=5,y=0,L=$3(d)$          ]{v50}
\Vertex[x=6,y=1,L=$\rho_{3\rightarrow 2}$]{v61}
\Vertex[x=7,y=1,L=$\rho_{2\rightarrow 1}$]{v71}
\Vertex[x=8,y=1,L=$\eta_{3,2}$           ]{v81}
\Vertex[x=9,y=1,L=$\oplus$           ]{v91}
\Vertex[x=9,y=0,L=$3(c)$          ]{v90}
\Vertex[x=10,y=1,L=$\eta_{2,1}$           ]{v101}
\Vertex[x=11,y=1,L=$\oplus$           ]{v111}
\Vertex[x=11,y=0,L=$2(b)$          ]{v110}
\Vertex[x=12,y=1,L=$1(a)$                 ]{v121}
\SetVertexNormal[MinSize=6pt]
\Vertex[x=4,y=1,L=$\eta_{3,2}$           ]{v41}
\Edge(v50)(v51)
\Edge(v90)(v91)
\Edge(v110)(v111)
\Edges(v41,v51,v61,v71,v81,v91,v101,v111,v121)

\end{tikzpicture}
\caption{The rooted tree representing a 3-expression for $P_4$.}\label{fig:cwtree}
\end{center}
\end{figure}

Let $G$ be a graph.
The {\it complement} of $G$, denoted by $\overline{G}$, has vertex set $V(\overline{G})=V(G)$ and an edge between two distinct vertices 
if and only if these vertices are not adjacent in~$G$.

Let $G$ be a graph. We define the following five operations.
The {\em contraction} of an edge $uv$ removes~$u$ and~$v$ from $G$, and
replaces them by a new vertex made adjacent to precisely those vertices that
were adjacent to $u$ or $v$ in $G$. By definition, edge contractions create neither self-loops nor multiple edges.
The \emph{subdivision} of an edge $uv$ replaces $uv$ by a new vertex~$w$ with edges $uw$ and~$vw$.
Let $u\in V(G)$ be a vertex that has exactly two neighbours $v,w$, and moreover let $v$ and $w$ be non-adjacent.
The \emph{vertex dissolution} of~$u$ removes~$u$ and adds the edge~$vw$.
For an induced subgraph $G'\ssi G$,  the {\em subgraph complementation} operation (acting on $G$ with respect to $G'$) replaces every edge present in $G'$ 
by a non-edge, and vice versa.  
Similarly, for two disjoint vertex subsets~$X$ and~$Y$ in~$G$, the {\em bipartite complementation} operation with respect to $X$ and~$Y$ acts on~$G$ by replacing
every edge with one end-vertex in~$X$ and the other one in $Y$ by a non-edge and vice versa.

We now state some useful facts for dealing with clique-width. We will use these facts throughout the paper.
Let $k\geq 0$ be a constant and let $\gamma$ be some graph operation. 
We say that a graph class ${\cal G'}$ is {\it $(k,\gamma)$-obtained} from a graph class ${\cal G}$ 
if the following two conditions hold:
\begin{enumeratei}
\item every graph in ${\cal G'}$ is obtained from a graph in ${\cal G}$ by performing $\gamma$ at most $k$ times, and
\item for every $G\in {\cal G}$ there exists at least one graph 
in ${\cal G'}$ obtained from $G$ by performing $\gamma$ at most~$k$ times.
\end{enumeratei}
If we do not impose a finite upper bound $k$ on the number of applications of $\gamma$ then we write  that ${\cal G}'$ is $(\infty,\gamma)$-obtained from ${\cal G}$.

We say that $\gamma$ {\em preserves} boundedness of clique-width if
for any finite constant~$k$ and any graph class ${\cal G}$, any graph class 
${\cal G}'$ that is $(k,\gamma)$-obtained from ${\cal G}$
has bounded clique-width if and only if  ${\cal G}$  has bounded clique-width.

\begin{enumerate}[{Fact} 1.]
\item \label{fact:del-vert}Vertex deletion preserves boundedness of clique-width~\cite{LR04}.\\[-1em]

\item \label{fact:comp}Subgraph complementation preserves boundedness of clique-width~\cite{KLM09}.\\[-1em]

\item \label{fact:bip}Bipartite complementation preserves boundedness of clique-width~\cite{KLM09}.\\[-1em]

\item \label{fact:subdiv}
For a class of graphs ${\cal G}$ of {\em bounded} maximum 
degree, let ${\cal G}'$ be a class of graphs that is $(\infty,{\tt es})$-obtained from ${\cal G}$, where ${\tt es}$ is the edge subdivision operation.
Then~${\cal G}$ has bounded clique-width if and only if  ${\cal G}'$ has
bounded clique-width~\cite{KLM09}.
\end{enumerate}

\medskip
\noindent
It is easy to show that the condition on the maximum degree in Fact~\ref{fact:subdiv} is necessary
for the reverse (i.e.  the ``only if'') direction:
for a graph $G$ of arbitrarily large clique-width, take a clique~$K$ (which has clique-width at most~2) 
with vertex set $V(K)=V(G)$, apply an edge subdivision 
on an edge~$uv$ in~$K$ if and only if $uv$ is not an edge in $G$ and, in order to obtain $G$ from this graph, remove any vertex introduced by 
an edge subdivision (this does not increase the clique-width).
As 
another 
aside, note that the reverse direction of Fact~\ref{fact:subdiv} also holds if we replace ``edge subdivisions'' by ``edge contractions''.\footnote{Combine the fact that a class of graphs of bounded maximum degree has bounded clique-width if and only if it has bounded tree-width~\cite{GW00} with the well-known fact that
 edge contractions do not increase the tree-width of a graph.}
It was an open problem~\cite{Gu07} whether the condition on maximum degree was also necessary in this case. 
This was recently solved by Courcelle~\cite{Co14}, who showed that if~${\cal G}$ is the class of graphs of clique-width~3 and~${\cal G}'$ is the class of graphs obtained from graphs in~${\cal G}$ by applying one or more edge contraction operations then~${\cal G}'$ has unbounded clique-width.

We also use a number of other elementary results on the clique-width of graphs.
The first one is well known (see e.g.~\cite{CO00}) and straightforward to check. 

\begin{lemma}\label{l-atmost2}
The clique-width of a graph with maximum degree at most~$2$ is at most~$4$.
\end{lemma}

We also need the well-known notion of a {\em wall}. We do not formally define this notion but instead refer to 
\figurename~\ref{f-walls}, in which three examples of walls of different height are depicted.
The class of walls is well known to have unbounded clique-width; see for example~\cite{KLM09}.
(Note that walls have maximum degree at most~3, hence the degree bound in Lemma~\ref{l-atmost2} is tight.)

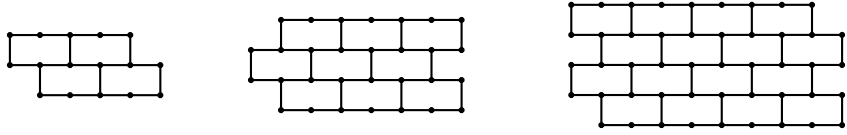
\begin{figure}
\begin{center}
\begin{minipage}{0.2\textwidth}
\centering
\begin{tikzpicture}[scale=0.4, every node/.style={scale=0.3}]
\GraphInit[vstyle=Simple]
\SetVertexSimple[MinSize=6pt]
\Vertex[x=1,y=0]{v10}
\Vertex[x=2,y=0]{v20}
\Vertex[x=3,y=0]{v30}
\Vertex[x=4,y=0]{v40}
\Vertex[x=5,y=0]{v50}

\Vertex[x=0,y=1]{v01}
\Vertex[x=1,y=1]{v11}
\Vertex[x=2,y=1]{v21}
\Vertex[x=3,y=1]{v31}
\Vertex[x=4,y=1]{v41}
\Vertex[x=5,y=1]{v51}

\Vertex[x=0,y=2]{v02}
\Vertex[x=1,y=2]{v12}
\Vertex[x=2,y=2]{v22}
\Vertex[x=3,y=2]{v32}
\Vertex[x=4,y=2]{v42}

\Edges(    v10,v20,v30,v40,v50)
\Edges(v01,v11,v21,v31,v41,v51)
\Edges(v02,v12,v22,v32,v42)

\Edge(v01)(v02)

\Edge(v10)(v11)

\Edge(v21)(v22)

\Edge(v30)(v31)

\Edge(v41)(v42)

\Edge(v50)(v51)

\end{tikzpicture}
\end{minipage}
\begin{minipage}{0.3\textwidth}
\centering
\begin{tikzpicture}[scale=0.4, every node/.style={scale=0.3}]
\GraphInit[vstyle=Simple]
\SetVertexSimple[MinSize=6pt]
\Vertex[x=1,y=0]{v10}
\Vertex[x=2,y=0]{v20}
\Vertex[x=3,y=0]{v30}
\Vertex[x=4,y=0]{v40}
\Vertex[x=5,y=0]{v50}
\Vertex[x=6,y=0]{v60}
\Vertex[x=7,y=0]{v70}

\Vertex[x=0,y=1]{v01}
\Vertex[x=1,y=1]{v11}
\Vertex[x=2,y=1]{v21}
\Vertex[x=3,y=1]{v31}
\Vertex[x=4,y=1]{v41}
\Vertex[x=5,y=1]{v51}
\Vertex[x=6,y=1]{v61}
\Vertex[x=7,y=1]{v71}

\Vertex[x=0,y=2]{v02}
\Vertex[x=1,y=2]{v12}
\Vertex[x=2,y=2]{v22}
\Vertex[x=3,y=2]{v32}
\Vertex[x=4,y=2]{v42}
\Vertex[x=5,y=2]{v52}
\Vertex[x=6,y=2]{v62}
\Vertex[x=7,y=2]{v72}

\Vertex[x=1,y=3]{v13}
\Vertex[x=2,y=3]{v23}
\Vertex[x=3,y=3]{v33}
\Vertex[x=4,y=3]{v43}
\Vertex[x=5,y=3]{v53}
\Vertex[x=6,y=3]{v63}
\Vertex[x=7,y=3]{v73}

\Edges(    v10,v20,v30,v40,v50,v60,v70)
\Edges(v01,v11,v21,v31,v41,v51,v61,v71)
\Edges(v02,v12,v22,v32,v42,v52,v62,v72)
\Edges(    v13,v23,v33,v43,v53,v63,v73)

\Edge(v01)(v02)

\Edge(v10)(v11)
\Edge(v12)(v13)

\Edge(v21)(v22)

\Edge(v30)(v31)
\Edge(v32)(v33)

\Edge(v41)(v42)

\Edge(v50)(v51)
\Edge(v52)(v53)

\Edge(v61)(v62)

\Edge(v70)(v71)
\Edge(v72)(v73)
\end{tikzpicture}
\end{minipage}
\begin{minipage}{0.35\textwidth}
\centering
\begin{tikzpicture}[scale=0.4, every node/.style={scale=0.3}]
\GraphInit[vstyle=Simple]
\SetVertexSimple[MinSize=6pt]
\Vertex[x=1,y=0]{v10}
\Vertex[x=2,y=0]{v20}
\Vertex[x=3,y=0]{v30}
\Vertex[x=4,y=0]{v40}
\Vertex[x=5,y=0]{v50}
\Vertex[x=6,y=0]{v60}
\Vertex[x=7,y=0]{v70}
\Vertex[x=8,y=0]{v80}
\Vertex[x=9,y=0]{v90}

\Vertex[x=0,y=1]{v01}
\Vertex[x=1,y=1]{v11}
\Vertex[x=2,y=1]{v21}
\Vertex[x=3,y=1]{v31}
\Vertex[x=4,y=1]{v41}
\Vertex[x=5,y=1]{v51}
\Vertex[x=6,y=1]{v61}
\Vertex[x=7,y=1]{v71}
\Vertex[x=8,y=1]{v81}
\Vertex[x=9,y=1]{v91}

\Vertex[x=0,y=2]{v02}
\Vertex[x=1,y=2]{v12}
\Vertex[x=2,y=2]{v22}
\Vertex[x=3,y=2]{v32}
\Vertex[x=4,y=2]{v42}
\Vertex[x=5,y=2]{v52}
\Vertex[x=6,y=2]{v62}
\Vertex[x=7,y=2]{v72}
\Vertex[x=8,y=2]{v82}
\Vertex[x=9,y=2]{v92}

\Vertex[x=0,y=3]{v03}
\Vertex[x=1,y=3]{v13}
\Vertex[x=2,y=3]{v23}
\Vertex[x=3,y=3]{v33}
\Vertex[x=4,y=3]{v43}
\Vertex[x=5,y=3]{v53}
\Vertex[x=6,y=3]{v63}
\Vertex[x=7,y=3]{v73}
\Vertex[x=8,y=3]{v83}
\Vertex[x=9,y=3]{v93}

\Vertex[x=0,y=4]{v04}
\Vertex[x=1,y=4]{v14}
\Vertex[x=2,y=4]{v24}
\Vertex[x=3,y=4]{v34}
\Vertex[x=4,y=4]{v44}
\Vertex[x=5,y=4]{v54}
\Vertex[x=6,y=4]{v64}
\Vertex[x=7,y=4]{v74}
\Vertex[x=8,y=4]{v84}

\Edges(    v10,v20,v30,v40,v50,v60,v70,v80,v90)
\Edges(v01,v11,v21,v31,v41,v51,v61,v71,v81,v91)
\Edges(v02,v12,v22,v32,v42,v52,v62,v72,v82,v92)
\Edges(v03,v13,v23,v33,v43,v53,v63,v73,v83,v93)
\Edges(v04,v14,v24,v34,v44,v54,v64,v74,v84)

\Edge(v01)(v02)
\Edge(v03)(v04)

\Edge(v10)(v11)
\Edge(v12)(v13)

\Edge(v21)(v22)
\Edge(v23)(v24)

\Edge(v30)(v31)
\Edge(v32)(v33)

\Edge(v41)(v42)
\Edge(v43)(v44)

\Edge(v50)(v51)
\Edge(v52)(v53)

\Edge(v61)(v62)
\Edge(v63)(v64)

\Edge(v70)(v71)
\Edge(v72)(v73)

\Edge(v81)(v82)
\Edge(v83)(v84)

\Edge(v90)(v91)
\Edge(v92)(v93)
\end{tikzpicture}
\end{minipage}
\caption{Walls of height 2, 3, and 4, respectively.}\label{f-walls}
\end{center}
\end{figure}

A \emph{$k$-subdivided wall} is a graph obtained from a wall after subdividing each edge 
exactly $k$ times for some constant $k\geq 0$.

The following lemma is well known and follows from combining Fact~\ref{fact:subdiv} with the aforementioned fact that walls 
have maximum degree at most~3 and unbounded clique-width.
 
\begin{lemma}[\cite{LR06}]\label{l-walls}
For any constant $k\geq 0$, the class of $k$-subdivided walls has unbounded clique-width.
\end{lemma}

For $r\geq 1$, the graphs $C_r$, $K_r$, $P_r$ denote the cycle, complete graph and path on $r$ vertices, respectively, and the graph $K_{1,r}$ denotes 
the star on $r+1$ vertices. The graph~$K_{1,3}$ is also called the {\em claw}.
For $1\leq h\leq i\leq j$, let $S_{i,j,k}$ denote the tree that has only one vertex~$x$ of degree $3$ and that has exactly three leaves, which are of distance $i$, $j$ and $k$ from~$x$, respectively.
Observe  that $S_{1,1,1}=K_{1,3}$. A graph $S_{i,j,k}$ is said to be a {\it subdivided claw}. We let ${\cal S}$ be the class of graphs each connected component of which is either a subdivided claw or a path. 

Like Lemma~\ref{l-atmost2}, the following lemma is also well known and follows from Lemma~\ref{l-walls}, by choosing appropriate values for $k$.
 
\begin{lemma}[\cite{LR06}]\label{l-classS}
Let $\{H_1,\ldots,H_p\}$ be a finite set of graphs. 
If $H_i\notin {\cal S}$ for $i=1,\ldots,p$ then the class of $(H_1,\ldots,H_p)$-free 
graphs has unbounded~clique-width.
\end{lemma}

We say that~$G$ is  \emph{bipartite} if its vertex set can be partitioned into two (possibly empty) independent sets $B$ and $W$. We say that $(B,W)$ is a {\em bipartition} of $G$.
Let $H$ be a bipartite graph with a fixed partition $(B_H,W_H)$.
A bipartite graph $G$ is {\em strongly} $H$-free if $G$ is $H$-free or else $G$ has no bipartition $(B_G,W_G)$  with  $B_H\subseteq B_G$ and $W_H\subseteq W_G$
such that $bw \in E(G)$ if and only if $bw \in E(H)$ for all $b \in B_H$ and $w \in W_H$.
Lozin and Volz~\cite{LV08} characterized all bipartite graphs $H$ for which the class of strongly $H$-free bipartite graphs 
has bounded clique-width.
Recently, we proved a similar characterization for $H$-free bipartite graphs; we will use this result in Section~\ref{sec:twographs}.

\begin{lemma}[\cite{DP14}]\label{l-bipartite}
Let $H$ be a graph. The class of $H$-free bipartite graphs has bounded
clique-width if and only if one of the following cases holds:
\begin{itemize}
\item $H=sP_1$ for some $s\geq 1$
\item $H\ssi K_{1,3}+3P_1$
\item $H\ssi K_{1,3}+P_2$ 
\item $H\ssi P_1+S_{1,1,3}$
\item $H\ssi S_{1,2,3}$.
\end{itemize}
\end{lemma}

From the same paper we will also need the following lemma.

\begin{lemma}[\cite{DP14}]\label{l-addit}
 Let $H\in {\cal S}$. Then $H$ is 
 $(2P_1+2P_2,\allowbreak 2P_1+P_4,\allowbreak 4P_1+P_2,\allowbreak3P_2,\allowbreak 2P_3)$-free
if and only if $H=sP_1$ for some integer $s\geq 1$ or $H$ is an induced subgraph of one of the graphs in 
$\{K_{1,3}+3P_1,K_{1,3}+P_2, P_1+S_{1,1,3}, S_{1,2,3}\}$.
\end{lemma}

We say that a graph $G$ is \emph{complete multipartite} if $V(G)$ can be partitioned into $k$ independent sets $V_1,\ldots,V_k$ for some integer $k$, such that two vertices are adjacent if and only if they belong to two different sets $V_i$ and $V_j$.
The next result is due to Olariu~\cite{Olariu88} (the graph $\overline{P_1+P_3}$ is also called the {\em paw}).

\begin{lemma}[\cite{Olariu88}]\label{l-olariu}
Every connected $(\overline{P_1+P_3})$-free graph is either complete multipartite or $K_3$-free.
\end{lemma}

Every complete multipartite graph has clique-width at most~2. Also, the definition of clique-width directly implies that the clique-width
of any graph is equal to the maximum clique-width of its connected components. Hence, Lemma~\ref{l-olariu} immediately implies the following 
(well-known) result.

\begin{lemma}\label{l-paw}
For any graph $H$, the class of $(\overline{P_1+P_3},H)$-free graphs has bounded clique-width
if and only if the class of $(K_3,H)$-free graphs has bounded clique-width.
\end{lemma}

\begin{sloppypar}
Kratsch and Schweitzer~\cite{KS12} proved that the {\sc Graph Isomorphism}
problem is graph-isomorphism complete for the class of
$(K_4,P_1+P_4)$-free graphs. It is a straightforward exercise to simplify their construction 
and use analogous arguments to prove that the class of $(K_4,P_1+P_4)$-free graphs 
has unbounded clique-width. 
Recall that Schweitzer~\cite{Sc15} proved
that any graph class that allows a so-called simple path encoding has unbounded clique-width, implying this result as a direct consequence.
\end{sloppypar}

\begin{lemma}[\cite{Sc15}]\label{l-k4-p1+p-4}
The class of $(K_4,P_1+P_4)$-free graphs has unbounded clique-width.
\end{lemma}

\section{New Classes of Bounded Clique-width}\label{sec:bounded}

In this section we identify two new graph classes that have bounded
clique-width, namely the classes of $(\overline{P_1+P_3},P_1+S_{1,1,2})$-free
graphs and $(\overline{P_1+P_3},\allowbreak K_{1,3}+\nobreak 3P_1)$-free graphs.

We first prove that the class of  $(\overline{P_1+P_3},P_1+S_{1,1,2})$-free graphs has bounded clique-width.
To do so we use a similar approach to that used by 
Dabrowski, Lozin, Raman and Ries~\cite{DLRR12} to prove that the classes of $(K_3,S_{1,1,3})$-free and
$(K_3,K_{1,3}+P_2)$-free graphs have bounded clique-width.

\begin{theorem}\label{thm:K_3s_112+P_1}
The class of $(\overline{P_1+P_3},P_1+S_{1,1,2})$-free graphs has bounded clique-width.
\end{theorem}

\begin{proof}
Let $G$ be a $(\overline{P_1+P_3},P_1+S_{1,1,2})$-free graph.
By Lemma~\ref{l-paw} we may assume~$G$ is $(K_3,\allowbreak P_1+\nobreak S_{1,1,2})$-free.
Without loss of generality, we may also assume that~$G$ is connected (as otherwise we could consider each connected component of~$G$ separately).
If~$G$ is bipartite, then $G$ has bounded clique-width by Lemma~\ref{l-bipartite}.
For the remainder of the proof we assume that $G$ is not bipartite, that is, $G$ contains an induced odd cycle $C=v_1v_2\cdots v_kv_1$. Because~$G$ is $K_3$-free, $k\geq 5$.

First, suppose that $k\geq 7$.
We claim that $G=C$. Indeed, suppose not. Since~$G$ is connected,~$G$ must have
a vertex $x\notin V(C)$ that is adjacent to a vertex of~$C$. Since $G$ is
$K_3$-free, $x$ cannot be adjacent to any two consecutive vertices of the cycle
$C$. Since $C$ is an odd cycle, $x$ must therefore have two consecutive non-neighbours on the cycle.
Without loss of generality we assume that~$x$ is
adjacent to $v_1$ and non-adjacent to $v_{k-1}$ and $v_k$.
Then $x$ must be adjacent to $v_4$, otherwise $G[v_1,x,v_2,v_k,v_{k-1},v_4]$ would be isomorphic to $P_1+S_{1,1,2}$. 
Now $x$ cannot be adjacent to $v_3$ or~$v_5$, since $G$ is $K_3$-free.
However, then $G[v_1,x,v_k,v_2,v_3,v_5]$ would be a $P_1+S_{1,1,2}$, which is a contradiction. Hence, $G=C$ and as such has 
clique-width at most~4 by Lemma~\ref{l-atmost2}.

From now on we assume that $k=5$.
Every vertex not on $C$ has at most two neighbours on the cycle, and if it has two, then
these neighbours on $C$ cannot be consecutive vertices of~$C$ 
(since $G$ is $K_3$-free).  
We now partition the vertices of $G$ not in $C$ into sets, depending on their
neighbourhood in~$C$. We let $X$ denote the vertices with no neighbours on the cycle. We let $V_i$ denote the set of all vertices not on the cycle $C$ that are adjacent to both $v_{i-1}$ and~$v_{i+1}$, where subscripts are interpreted modulo~5.
We let $W_i$ denote the set of all vertices that are adjacent to~$v_i$ but to no other vertices of $C$. We say
that a set $V_i$ or $W_i$ is {\em large} if it contains at least 
two  vertices, otherwise we say that it is {\em small}. We say that a set in $\{V_i,W_i\}$ and a set in $\{V_j,W_j\}$ 
are {\em consecutive} if $v_i$ and $v_j$ are consecutive vertices on $C$, otherwise, we say that they are {\em opposite}.
Note that each $V_i$ and each $W_i$ is an independent set, since $G$ is
$K_3$-free. We now investigate the possible adjacencies between vertices of these sets through a
series of eight claims.

\begin{enumerate}
\item {\em $X$ is an independent set and every vertex in $X$ is adjacent to
every vertex in~$V_i$ and~$W_i$.} Suppose there is a vertex $x\in X$.  Since
$G$ is connected, there must be a vertex $y \not \in V(C)$ with a neighbour on
the cycle. We may assume without loss of generality that $y$ is adjacent to
$v_1$, but not to $v_2,v_3$ or $v_5$. Then $x$ must be adjacent to $y$,
otherwise $G[v_1,y,v_5,v_2,v_3,x]$ would be isomorphic to $P_1+S_{1,1,2}$.
Hence every vertex in $X$ is adjacent to every vertex in $V_i$ and $W_i$ for
all $i$. Because of the fact that if $X$ is non-empty then some $V_i$ or $W_i$
must also be non-empty and the fact that~$G$ is $K_3$-free, $X$ must be an
independent set.

\item {\em If $V_i$ and $V_j$ are opposite then no vertex of~$V_i$ is adjacent to a vertex of~$V_j$.} This
follows from the fact that any two such vertices have a common neighbour on $C$ and the fact that $G$ is $K_3$-free.

\item {\em If $V_i$ and $V_j$ are consecutive and large then every vertex of $V_i$
is adjacent to every vertex of $V_j$.} Without loss of generality, let
$i=1,j=2$.  Suppose $y \in V_1$ is not adjacent to $z_1,z_2\in V_2$. Then
$G[v_1,z_1,z_2,v_2,y,v_4]$ is a $P_1+S_{1,1,2}$. Now suppose that $y$ is
adjacent to $z_1$, but not to~$z_2$, then $G[y,v_2,z_1,v_5,v_4,z_2]$ is isomorphic to
$P_1+S_{1,1,2}$, which is a contradiction.

\item {\em If $V_i$ and $W_j$ are consecutive then one of them must be empty.}
Suppose, for contradiction, that there exist vertices $x \in V_1$ and~$y \in W_2$. Then $x$ and~$y$
are non-adjacent, as $G$ is $K_3$-free. However, 
then $G[v_5,v_1,x,v_4,v_3,y]$ is isomorphic to $P_1+S_{1,1,2}$, which is a contradiction.

\item {\em If $V_i$ and $W_j$ are opposite and $W_j$ is large then no vertex of
$V_i$ has a neighbour in $W_j$.} Let $y \in V_1$ and $z_1,z_2 \in W_3$.
If $y$ is adjacent to both $z_1$ and $z_2$, then $G[y,z_1,z_2,v_2,v_1,v_4]$ is
isomorphic to $P_1+S_{1,1,2}$. So $y$ is adjacent to at most one vertex of $W_3$, say
$y$ is adjacent to $z_1$, but not to $z_2$. Then
$G[v_5,v_1,v_4,y,z_1,z_2]$ is isomorphic to $P_1+S_{1,1,2}$, which is a contradiction.

\begin{sloppypar}
\item {\em Every vertex in $V_i$ has at most one non-neighbour in~$W_i$ and
vice versa.} If $y_1 \in V_1$  has two non-neighbours $z_1, z_2 \in W_1$ then
the graph $G[v_1,z_1,z_2,v_2,y_1,v_4]$ is isomorphic to $P_1+S_{1,1,2}$, which
is a contradiction.  If $z_1 \in W_1$ has two non-neighbours $y_1,y_2 \in V_1$
then $G[v_2,y_1,y_2,v_1,z_1,v_4]$ is isomorphic to $P_1+S_{1,1,2}$, which is
again a contradiction.

\item {\em If $W_i$ and $W_j$ are consecutive and $W_j$ is large then $W_i$ is
empty.} Without loss of generality, let $i=\nobreak 1$ and $j=2$. Suppose, for
contradiction, that~$y \in W_1$ and $z_1,z_2\in W_2$.  If $y$ is adjacent to
both~$z_1$ and $z_2$ then $G[y,z_1,z_2,v_1,v_5,v_3]$ is isomorphic to
$P_1+S_{1,1,2}$. Without loss of generality, we therefore assume that $y$ is
not adjacent to $z_1$. If $y$ is not adjacent to $z_2$ then
$G[v_2,z_1,z_2,v_1,y,v_4]$ is isomorphic to $P_1+S_{1,1,2}$. If $y$ is adjacent
to $z_2$, then $G[v_2,v_3,z_1,z_2,y,v_5]$ is isomorphic to $P_1+S_{1,1,2}$.
Hence in all three cases we have a contradiction.
\end{sloppypar}

\item {\em If $W_i$ and $W_j$ are opposite then every vertex of $W_i$ must be
adjacent to every vertex of~$W_j$.} Without loss of generality, let $i=1$, $j=3$, 
$x \in W_1$, and $y \in W_3$. If $x$ and $y$ are not adjacent, then
$G[v_1,v_2,x,v_5,v_4,y]$ is isomorphic to $P_1+S_{1,1,2}$, which is not possible.
\end{enumerate}

We now do as follows. First, we remove 
the vertices of $C$ and all small sets~$V_i$ or~$W_i$ if they exist. In this way we remove at most $5+5+5=15$ vertices. Hence,~$G$ has bounded clique-width if and only 
if the resulting graph~$G'$ has bounded clique-width, by Fact~\ref{fact:del-vert}.
We then consider the remaining sets~$X$,~$V_i$ and~$W_i$ in~$G'$. 
We complement the edges between the vertices in $X$ and the vertices not in~$X$.
If~$V_i$ and~$V_j$ are consecutive, we complement the edges between them. 
If~$W_i$ and~$W_j$ are opposite, we complement the edges between them. 
Finally, for any pair $V_i$ and~$W_i$, we complement the edges between them.
Then $G'$ has bounded clique-width if and only if the resulting graph $G^*$ has bounded clique-width, by Fact~\ref{fact:bip}. 
If two vertices are adjacent in $G^*$, then they must be members of
some~$V_i$ and~$W_i$, respectively. By construction, $G^*[V_i \cup W_i]$ is a (not necessarily perfect) matching. Thus $G^*$ has clique-width at most~2, completing the proof.\qed
\end{proof}

Next, we prove that the class of $(\overline{P_1+P_3},K_{1,3}+3P_1)$-free graphs has bounded clique-width.
To do so we first prove Lemma~\ref{l:K_3s_111+2P_1}, which says that the class of
$(\overline{P_1+P_3},\allowbreak K_{1,3}+\nobreak 2P_1)$-free graphs has bounded clique-width.
We then use this
result to prove Theorem~\ref{thm:K_3s_111+3P_1}, which says that the larger
class of $(\overline{P_1+P_3},\allowbreak K_{1,3}+\nobreak 3P_1)$-free graphs also has bounded clique-width.
It is also possible to prove Theorem~\ref{thm:K_3s_111+3P_1} by
combining very similar arguments to those in the proof of 
Lemma~\ref{l:K_3s_111+2P_1} together with the fact that the class of
$(\overline{P_1+P_3},K_{1,3}+P_1)$-free graphs has bounded clique-width (which follows from
Theorem~\ref{thm:K_3s_112+P_1}).
However, we believe that such a combined proof would be much
harder to follow. 

\begin{lemma}\label{l:K_3s_111+2P_1}
The class of $(\overline{P_1+P_3},K_{1,3}+2P_1)$-free graphs has bounded clique-width.
\end{lemma}

\begin{proof}
Let~$G$ be a $(\overline{P_1+P_3},K_{1,3}+2P_1)$-free graph. By
Lemma~\ref{l-paw}, we may assume~$G$ is $(K_3,\allowbreak K_{1,3}+\nobreak 2P_1)$-free. Let $x$ be
an arbitrary vertex
in~$G$. Let $N_1=N(x)$ and $N_2=V(G)\setminus N[x]$. 
Since~$G$ is $K_3$-free, $N_1$ must be an independent set.
Since~$G$ is $(K_{1,3}+2P_1)$-free,~$G[N_2]$ must be $(K_{1,3}+P_1)$-free. 
Then~$G[N_2]$ must have bounded clique-width by 
Theorem~\ref{thm:K_3s_112+P_1}.

Suppose that $|N_1|\leq 2$. Then we delete $x$ and the vertices of $N_1$ and obtain 
a graph of bounded clique-width, namely $G[N_2]$.
By Fact~\ref{fact:del-vert}, we find that $G$ also has bounded clique-width.
Hence we may assume that $|N_1|\geq 3$.

We prove the following claim.

\medskip
\phantomsection\refstepcounter{ctrclaim}
\noindent
\displaycase{Claim \thectrclaim.}
\label{l:K_3s_111+2P_1r:clm1}
 {\em Let $S\subseteq N_2$ with $|S|\leq k$ for some $k$.  If
$G[N_2\setminus S]$ is complete bipartite, then the clique-width of $G$ is
bounded by a function of $k$. In particular, this includes the case where
$G[N_2\setminus S]$ is an independent set.}

\medskip
\noindent
To prove Claim~\ref{l:K_3s_111+2P_1r:clm1}, suppose that $G[N_2\setminus S]$ is complete bipartite. 
No vertex in $N_1$ has a neighbour in both partition classes of $G[N_2\setminus S]$, due to the fact that~$G$ is $K_3$-free.
Because $N_1$ is an independent set, this means that $G[N_1\cup (N_2\setminus S)]$ is bipartite, in addition to being  $(K_3,K_{1,3}+2P_1)$-free.
Hence, $G[N_1\cup (N_2\setminus S)]$ has bounded clique-width by 
Lemma~\ref{l-bipartite}.
Then by Fact~\ref{fact:del-vert}, $G=G[N_1\cup (N_2\setminus S)\cup S\cup \{x\}]$ has clique-width bounded by some function of $|S|$. This proves Claim~\ref{l:K_3s_111+2P_1r:clm1}.

\medskip
\noindent
We will use Claim~\ref{l:K_3s_111+2P_1r:clm1} later in the proof and now proceed as follows.
We fix three arbitrary vertices $x_1,x_2,x_3 \in N_1$; such vertices exist because $|N_1|\geq 3$.
Let $y_1,y_2,y_3$ be three arbitrary vertices of~$N_2$. We will show that at least one of them is adjacent to at least one of $x_1$, $x_2$,~$x_3$.
Because $G$ is $K_3$-free, two of $y_1$, $y_2$, $y_3$ are not pairwise adjacent, say $y_1y_2\notin E(G)$. If both $y_1$ and $y_2$ have no neighbour in $\{x_1,x_2,x_3\}$,
then $G[x,x_1,x_2,x_3,y_1,y_2]$ is isomorphic to $K_{1,3}+2P_1$, a contradiction. Hence, 
all vertices of $N_2$ except at most two have at least one neighbour in $\{x_1,x_2,x_3\}$. 
Then, by Fact~\ref{fact:del-vert}, we may assume without loss of generality that all vertices of~$N_2$ have at least one neighbour in $\{x_1,x_2,x_3\}$.

Let $A$ consist of those vertices of $N_2$ that are adjacent to $x_1$.
Let $B$ consist of those vertices of $N_2$ that are adjacent to $x_2$ but not to $x_1$.
Let $C=N_2\setminus (A\cup B)$. Note that every vertex in $C$ is adjacent to $x_3$ but not to $x_1$ or $x_2$.
Moreover, $A, B, C$ are three independent sets due to the fact that~$G$ is $K_3$-free.
If $C$ contains at least three vertices, say $c_1,c_2,c_3$, then $G[x_3,c_1,c_2,c_3,x_1,x_2]$ is isomorphic to
$K_{1,3}+2P_1$. Thus $|C|\leq 2$.
If $|A|\leq 7$, then $|A\cup C|\leq 9$. Moreover, $G[N_2\setminus (A\cup C)]= G[B]$ is complete bipartite, because $B$ is an independent set.
Hence, we may apply Claim~\ref{l:K_3s_111+2P_1r:clm1}. From now on we assume that $|A|\geq 8$, and similarly, that~$|B|\geq 8$.

At least one vertex of any pair from $B$ must be adjacent to at least one vertex of any triple from $A$; otherwise these five vertices, together with $x_1$,
induce a subgraph isomorphic to $K_{1,3}+2P_1$, since $A$ and $B$ are independent sets and $x_1$ is adjacent to all vertices of $A$ and to none of $B$. 
Fix three vertices $a_1,a_2,a_3\in A$. Then at most one vertex of $B$ has no neighbours in $\{a_1,a_2,a_3\}$. 
Because $|B|\geq 8$, this means that at least one of $a_1,a_{2},a_{3}$ must have at least three neighbours in~$B$.
By repeating this argument with different choices of $a_{1},a_{2},a_{3}$, we
find that all but at most two vertices in $A$ have at least three neighbours in $B$.
So, at least six vertices in $A$ have at least three neighbours in $B$, and vice versa.

Let $a\in A$ be adjacent to at least three vertices $b_1,b_2,b_3$ of $B$. If $a$ is not adjacent to some $b_4\in B$, then
$G[a_1,b_1,b_2,b_3,b_4,x]$ is isomorphic to
$K_{1,3}+2P_1$. Hence, every vertex of~$A$ with at least three neighbours in $B$ is adjacent to all vertices of $B$. 
By reversing the roles of $A$ and $B$, we find that every vertex in~$B$ with at least three neighbours in $A$ must be adjacent to all vertices of $A$.
Because there are at least  six vertices in $A$ with at least three neighbours in $B$, and vice versa, we conclude that all vertices of $A$ are adjacent to all vertices of $B$, that is,
$G[N_2\setminus C]=G[A\cup B]$ is complete bipartite. Because $|C|\leq 2$, we may apply Claim~\ref{l:K_3s_111+2P_1r:clm1} to complete the proof.\qed
\end{proof}

\begin{theorem}\label{thm:K_3s_111+3P_1}
The class of $(\overline{P_1+P_3},K_{1,3}+3P_1)$-free graphs has bounded clique-width.
\end{theorem}

\begin{proof}
Let $G$ be a $(\overline{P_1+P_3},K_{1,3}+3P_1)$-free graph. By Lemma~\ref{l-paw}, we may assume~$G$ is $(K_3,\allowbreak K_{1,3}+\nobreak 3P_1)$-free.
Suppose that $G$ contains a vertex of degree at most~18. If we remove this vertex and its neighbours, we obtain
a $(K_3,K_{1,3}+2P_1)$-free graph, which has bounded clique-width by Lemma~\ref{l:K_3s_111+2P_1}.
Hence,~$G$ also has bounded clique-width, by Fact~\ref{fact:del-vert}.
From now on we assume that $G$ has minimum degree at least~19 (the reason for choosing this number becomes clear later).

Let $x\in V(G)$. Let $N_1=N(x)$ and $N_2=V(G)\setminus N[x]$. Note that $|N_1|\geq 19$ and fix three arbitrarily-chosen vertices $x_1,x_2,x_3 \in N_1$. 
Let $Y$ be the set of vertices in $N_2$ that have no neighbour in~$\{x_1,x_2,x_3\}$.
We will need the following claim.

\medskip
\phantomsection\refstepcounter{ctrclaim}
\noindent
\displaycase{Claim \thectrclaim.}
\label{thm:K_3s_111+3P_1:clm1}
{\em $|Y|\leq 5$.}

\medskip
\noindent
We prove Claim~\ref{thm:K_3s_111+3P_1:clm1} as follows.
Suppose that there are three vertices $y_1,y_2,y_3 \in N_2$ that are pairwise non-adjacent.  
Then at least one of $y_1,y_2,y_3$ must be adjacent to at least one of $x_1,x_2,x_3$, as otherwise $G[x,x_1,x_2,x_3,y_1,y_2,y_3]$ would be isomorphic to $K_{1,3}+3P_1$.
Hence~$G[Y]$ is $3P_1$-free. Because~$G[Y]$ is also $K_3$-free, we apply Ramsey's Theorem and find that $|Y|\leq R(3,3)-1=6-1=5$.
This proves Claim~\ref{thm:K_3s_111+3P_1:clm1}.

\medskip
\noindent
We proceed as follows. Let $N_2'=N_2\setminus Y$.
Let $A$ consist of those vertices of $N_2'$ that are adjacent to~$x_1$.
Let $B$ consist of those vertices of $N_2'$ that are adjacent to~$x_2$ but not to $x_1$.
Let $C=N_2'\setminus (A\cup B)$. Note that every vertex in $C$ is adjacent to $x_3$, but not to $x_1$ or $x_2$.
Moreover, $A$, $B$, $C$ are three independent sets due to the fact that~$G$ is $K_3$-free.

We need the following claim.

\medskip
\medskip
\phantomsection\refstepcounter{ctrclaim}
\noindent
\displaycase{Claim \thectrclaim.}
\label{thm:K_3s_111+3P_1:clm2}
{\em Let $S,T\in \{A,B,C\}$ with $S \neq T$, $|S|\geq 9$ and $|T|\geq 9$.
Then there exist vertices $s\in S$ and $t\in T$ such that $G[(S\setminus \{s\})\cup (T\setminus \{t\})]$
is a complete bipartite graph minus a matching.}

\medskip
\noindent
We prove Claim~\ref{thm:K_3s_111+3P_1:clm2} as follows.
Suppose $S=A$ and $T=B$ with $|A|\geq 9$ and $|B|\geq 9$. 
Let $a,a',a''\in\nobreak A$ and $b,b',b''\in B$ be pairwise distinct.
Recall that $A$ and $B$ are independent sets.
Then at least one of $a,a',a''$ must be adjacent to at least one of $b,b',b''$, as otherwise the graph $G[x_1,a,a',a'',b,b',b'']$ would be isomorphic to 
$K_{1,3}+3P_1$. This means that at most two vertices in $B$ have no neighbour in $\{a,a',a''\}$. Hence, as $|B|\geq 9$,
at least one of $a,a',a''$ has at least three neighbours in $B$.
Repeating this argument with different choices of $a,a',a''$, we
find that all but at most two vertices in $A$ have at least three neighbours in $B$.

\begin{sloppypar}
Every vertex $a' \in A$ that is adjacent to at least three vertices of $B$,
say $b_1,b_2,b_3$, must be adjacent to all but at most one vertex of $B$, since
if $a'$ is not adjacent to $b_4,b_5 \in B$, then $G[a',b_1,b_2,b_3,x,b_4,b_5]$
would be a $K_{1,3}+3P_1$.  Because  all but at most two vertices in $A$
have at least three neighbours in~$B$, this means that all but at most two
vertices of $A$ are adjacent to all but at most one vertex of $B$.  Because
$|A|\geq 9> 7$, this means that every vertex of $B$ except at most one has at
least three neighbours in $A$. Let $b\in B$ be this exceptional vertex; if it
does not exist then we pick $b\in B$ arbitrarily.  If $b' \in B\setminus
\{b\}$, let $a_1,a_2,a_3$ be three of its neighbours in $A$. Then $b'$ cannot
be non-adjacent to two vertices, say $a_4,a_5$ in $A$, otherwise
$G[b',a_1,a_2,a_3,x,a_4,a_5]$ would be a $K_{1,3}+3P_1$. Thus every vertex
in $B\setminus \{b\}$ is adjacent to all but at most one vertex of $A$.  Since
$|B\setminus \{b\}| \geq 8 > 5$, every vertex in $A$, except at most one has at
least three neighbours in $B\setminus \{b\}$ and as stated above must therefore
be adjacent to all but at most one vertex of $B$. 
We let $a\in A$ denote this exceptional vertex; if it does not exist, then we pick $a\in A$ arbitrarily.
Because $A$ and $B$ are independent sets,  we conclude that $G[(A\setminus \{a\}) \cup (B\setminus \{b\})]$ is a complete bipartite graph minus a 
(not necessarily perfect) matching.
If a different pair of sets in $\{A,B,C\}$ both have at least nine vertices, the claim follows by the same arguments.
\end{sloppypar}

\medskip
\noindent
We now consider three different cases.

\medskip
\noindent
\displaycase{Case 1.} {\em At least two sets out of $A,B,C$ have less than nine vertices.}\\
Suppose  $|A|\leq 8$ and $|B|\leq 8$. Recall that $C, N_1$ are independent sets and that $G$ is $(K_{1,3}+3P_1)$-free.
Then $G[V(G)\setminus (\{x\}\cup A\cup B\cup Y)]=G[C\cup N_1]$ is bipartite and $(K_{1,3}+3P_1)$-free. Consequently, it has bounded clique-width by 
Lemma~\ref{l-bipartite}.
We have $|Y|\leq 5$ by Claim~\ref{thm:K_3s_111+3P_1:clm1}. Then $|\{x\}\cup A\cup B\cup Y|\leq 1+8+8+5=22$. Hence,~$G$ has bounded clique-width by Fact~\ref{fact:del-vert}.
If a different pair of sets in $\{A,B,C\}$ both have less than nine vertices, we apply the same arguments.

\medskip
\noindent
\displaycase{Case 2.} {\em Exactly one set out of $A,B,C$ has less than nine vertices.}\\
Suppose $|C|\leq 8$. Hence $|A|\geq 9$ and $|B|\geq 9$.
By Claim~\ref{thm:K_3s_111+3P_1:clm2} we find that there exist two vertices $a\in A$ and $b\in B$ such that $G[(A\setminus \{a\}) \cup (B\setminus \{b\})]$ is a complete bipartite graph minus a matching.
Let $x'\in\nobreak N_1$.
 Suppose, for contradiction, that $x'$ is adjacent to a vertex $a'\in A\setminus \{a\}$ and to a vertex $b'\in B\setminus \{b\}$. Then $x'$ is not adjacent to any other vertices of $(A\setminus \{a\}) \cup (B\setminus \{b\})$, 
 otherwise $G$ would not be $K_3$-free.
Recall that~$N_1$ is an independent set.
Hence 
$N(x')\subseteq \{a,b,a',b',x\}\cup C\cup Y$. We have $|Y|\leq 5$ by
Claim~\ref{thm:K_3s_111+3P_1:clm1}. Hence, $|N(x')|\leq 5+8+5=18$, which is a
contradiction since~$G$ has minimum degree at least~19.
We conclude that no vertex in $N_1$ has neighbours in both $A\setminus \{a\}$ and~$B\setminus \{b\}$. Because~$N_1$ is independent and $G$ is $(K_{1,3}+3P_1)$-free, this means that
$G[V(G)\setminus (\{a,b,x\}\cup C\cup\nobreak Y)]=\allowbreak G[N_1\cup (A\setminus \{a\}) \cup (B\setminus \{b\})]$ is bipartite and $(K_{1,3}+3P_1)$-free.
Consequently, it has bounded clique-width by Lemma~\ref{l-bipartite}.
Because $|\{a,b,x\} \cup C\cup Y|\leq 3+8+5=16$, we conclude that~$G$ has bounded clique-width by Fact~\ref{fact:del-vert}.
If $|A|\leq 8$ or $|B|\leq 8$, we repeat the above arguments with $A$ and~$B$ replaced by $B$ and $C$, or~$A$ and $C$, respectively.

\medskip
\noindent
\displaycase{Case 3.} {\em None of the sets $A,B,C$ has less than nine vertices.}\\
By Claim~\ref{thm:K_3s_111+3P_1:clm2}, we find that there exist vertices $a,a',b,b',c,c'$ such that  $G[(A\setminus\nobreak \{a\}) \cup (B\setminus \{b\})]$, $G[(A\setminus \{a'\}) \cup (C\setminus \{c\})]$, and $G[(B\setminus \{b'\}) \cup (C\setminus \{c'\})]$ are complete bipartite graphs minus a matching.
Hence $G[(A\setminus \{a,a'\}) \cup (B\setminus \{b,b'\})]$, $G[(A\setminus\nobreak \{a,a'\}) \cup (C\setminus \{c,c'\})]$, and $G[(B\setminus \{b,b'\}) \cup (C\setminus \{c,c'\})]$ are also complete bipartite graphs minus a matching.
Because $|A|\geq 9> 2$, $|B|\geq 9> 3$ and $|C|\geq 9> 4$, there exist vertices  $a_1\in A\setminus \{a,a'\}$, $b_1,b_2\in B\setminus \{b,b'\}$ and $c_1,c_2,c_3\in C\setminus \{c,c'\}$. 
Then $a_1$ is adjacent
to at least one of $b_1,b_2$  and to at least two of $c_1,c_2,c_3$. Moreover,
$b_1$ and $b_2$ are each adjacent to at least two of $c_1, c_2, c_3$. Hence $G$ is not $K_3$-free. This contradiction completes the proof.\qed
\end{proof}

\section{New Classes of Unbounded Clique-width}\label{sec:unbounded}

In order to prove our results, we
first present a general construction for obtaining graph classes of unbounded clique-width. 
We then show how we can use our construction to obtain two new classes of unbounded clique-width.
Our construction generalizes the constructions used by Golumbic and Rotics~\cite{GR99b},\footnote{The class of 
 (square) grids was first shown to have unbounded clique-width by
Makowsky and Rotics~\cite{MR99}.  The construction of~\cite{GR99b} determines the
exact clique-width of square grids and narrows the clique-width of non-square grids to two values.} 
Brandst\"adt et al.~\cite{BELL06} and Lozin and Volz~\cite{LV08} to prove that the classes of square grids, $K_4$-free co-chordal graphs 
and $2P_3$-free graphs, respectively, have unbounded clique-width. It can also be used to show directly that the classes of $k$-subdivided walls have unbounded clique-width (Lemma~\ref{l-walls}).

\begin{theorem}\label{thm:generalunbounded}
For $m\geq 0$ and $n > m+1$ the clique-width of a graph $G$ is at least 
$\lfloor\frac{n-1}{m+1}\rfloor+1$
if~$V(G)$ has a partition into sets $V_{i,j} (i,j \in \{0,\ldots,n\})$ with 
the following properties:
\begin{enumerate}
\item \label{prop:v_i0-small} $|V_{i,0}| \leq 1$ for all $i\geq 1$.\\[-9pt]
\item \label{prop:v_0j-small} $|V_{0,j}| \leq 1$ for~all~$j\geq 1$.\\[-9pt]
\item \label{prop:v_ij-nonempty} $|V_{i,j}|\geq 1$ for all $i,j\geq 1$.\\[-9pt]
\item \label{prop:row-connected} $G[\cup^n_{j=0}V_{i,j}]$ is connected for all $i\geq 1$.\\[-9pt]
\item \label{prop:column-connected} $G[\cup^n_{i=0}V_{i,j}]$ is connected for all $j\geq 1$.\\[-9pt]
\item \label{prop:v_k0-nbrs} For $i,j,k\geq 1$, if a vertex of $V_{k,0}$ is adjacent to a vertex of $V_{i,j}$ then $i \leq k$.\\[-9pt]
\item \label{prop:v_0k-nbrs} For $i,j,k\geq 1$, if a vertex of $V_{0,k}$ is adjacent to a vertex of $V_{i,j}$ then $j \leq k$.\\[-9pt]
\item \label{prop:v_ij-nbrs} For $i,j,k,\ell\geq 1$, if a vertex of $V_{i,j}$ is adjacent to a vertex of $V_{k,\ell}$ then $|k-i|\leq m$ and $|\ell-j| \leq m$.\\[-9pt]
\end{enumerate}
\end{theorem}

\begin{proof}
Fix integers $n,m$ with $m\geq 0$ and $n > m+1$, and let $G$ be a graph with a partition as described above.  
For $i>0$ we let $R_i=\cup^n_{j=0}V_{i,j}$ be a {\em row} of $G$ and for $j>0$ we let
$C_j=\cup^n_{i=0}V_{i,j}$ be a {\em column} of $G$.  Note that $G[R_i]$ and
$G[C_j]$ are non-empty by Property~\ref{prop:v_ij-nonempty}. They are connected graphs by Properties~\ref{prop:row-connected} and~\ref{prop:column-connected}, respectively.

Consider a $k$-expression for $G$. We will show that $k \geq
\lfloor\frac{n-1}{m+1}\rfloor+1$.
As stated in Section~\ref{sec:prelim}, this $k$-expression can be
represented by a rooted tree $T$, whose leaves correspond to the operations of
vertex creation and whose internal nodes correspond to the other three
operations (see \figurename~\ref{fig:cwtree} for an example). 
We denote the subgraph of $G$ that corresponds to
the  subtree of $T$ rooted at node $x$ by~$G(x)$.
Note that $G(x)$ may not be an induced subgraph of~$G$ as missing edges can be added by operations 
corresponding to $\eta_{i,j}$ nodes higher up in $T$. 

Let $x$ be a deepest (i.e. furthest from the root) $\oplus$ node in $T$ such that $G(x)$ contains an entire row or
an entire column of $G$ (the node $x$ may not be unique). Let $y$ and $z$ be the children of~$x$ in~$T$. Colour all vertices
in~$G(y)$ blue and all vertices in $G(z)$ red. Colour all remaining vertices of $G$
yellow. Note that a vertex of $G$ appears in $G(x)$ if and only if it is
coloured either red or blue and that there is no edge in $G(x)$ between a red and a blue vertex.
Due to our choice of $x$, $G$ contains a row or a column none of whose vertices are
yellow, but no row or column of $G$ is entirely blue or entirely red.
Without loss of generality, assume that $G$ contains a non-yellow column.

Because $G$ contains a non-yellow column, each row of~$G$ contains a non-yellow vertex, by Property~\ref{prop:v_ij-nonempty}.
Since no row is entirely red or entirely blue, every row of~$G$ is therefore coloured
with at least two colours. Let $R_i$ be an arbitrary row.
Since $G[R_i]$ is connected, there must be two
adjacent vertices $v_i,w_i \in R_i$ in~$G$, such that $v_i$ is either red or blue and
$w_i$ has a different colour than $v_i$. Note that~$v_i$ and $w_i$ are therefore not adjacent in $G(x)$
(recall that if $w_i$ is yellow then it is not even present as a vertex of $G(x)$).

Now consider indices $i,k\geq 1$ with $k>i+m$.  
By Properties~\ref{prop:v_k0-nbrs} and~\ref{prop:v_ij-nbrs}, no vertex of $R_i$ is adjacent to
a vertex of $R_{k}\setminus V_{k,0}$ in $G$.  Therefore, since $|V_{k,0}|\leq
1$ by Property~\ref{prop:v_i0-small}, we conclude that either $v_i$ and $w_i$  are not  adjacent to $v_k$ in $G$, or 
$v_i$ and~$w_i$ are not adjacent to $w_k$ in $G$.
In particular, this implies that $w_i$ is not adjacent to~$v_k$ in $G$ or that $w_k$ is not adjacent to~$v_i$ in $G$.
Recall that~$v_i$ and $w_i$ are adjacent in $G$ but not in $G(x)$, and the same holds for $v_k$ and $w_k$.
Hence, a~$\eta_{i,j}$ node higher up in the tree, makes $w_i$ adjacent to $v_i$ but not to $v_k$, or makes 
$w_k$ adjacent to~$v_k$ but not to $v_i$.
This means that $v_i$ and $v_k$ must have different labels in~$G(x)$.
We conclude that
$v_1,v_{(m+1)+1},v_{2(m+1)+1},v_{3(m+1)+1},\ldots,v_{\left(\left\lfloor\frac{n-1}{m+1}\right\rfloor\right)(m+1)+1}$
must all have different labels in $G(x)$. Hence, the $k$-expression of $G$ uses at least
$\lfloor\frac{n-1}{m+1}\rfloor+1$ 
labels.\qed
\end{proof}

We now use Theorem~\ref{thm:generalunbounded} to determine two new graph classes 
that have unbounded clique-width. 

\newpage
\begin{theorem}\label{thm:P_6diamond}
The class of $(P_6,\overline{2P_1+P_2})$-free graphs has unbounded
clique-width.
\end{theorem}

\begin{proof}
Let $n\geq 1$ be an integer. 
Using the notation of Theorem~\ref{thm:generalunbounded}, we construct a
graph~$G_n$ as follows. We define vertex subsets
\begin{align*}
V_{0,0}&=\emptyset\\
V_{i,0}&=\{b_i\}\; \text{for}\; i\geq 1\\
V_{0,j}&=\{w_j\}\; \text{for}\; j\geq 1\\
V_{i,j}&=\{b_{i,j},r_{i,j},w_{i,j}\}\; \text{for}\; i,j\geq 1.
\end{align*}
We define edge subsets 
\begin{align*}
E_1&=\{b_{i,j}r_{i,j}, r_{i,j}w_{i,j}\; |\; i,j \in \{1,\ldots,n\}\}\\
E_2&=\{b_kw_{i,j}\; |\; i,j,k \in \{1,\ldots,n\}, i\leq k\}\\
E_3&=\{w_kb_{i,j}\; |\; i,j,k \in \{1,\ldots,n\}, j\leq k\}.
\end{align*}
Let $V(G_n)$ be the union of the sets $V_{i,j}$ for
$i,j \in \{0,\ldots,n\}$, and let $E(G_n)=E_1\cup E_2 \cup E_3$. By
Theorem~\ref{thm:generalunbounded} with 
$m=0$, the graph $G_n$ has clique-width at least~$n$.

We now define the sets 
\begin{align*}
B_1&=\{b_i\; |\; i \in \{1,\ldots,n\}\}\\
W_1&=\{w_j\; |\; j \in \{1,\ldots,n\}\}\\
B_2&=\{b_{i,j}\; |\; i,j \in \{1,\ldots,n\}\}\\
R_2&=\{r_{i,j}\; |\; i,j \in \{1,\ldots,n\}\}\\
W_2&=\{w_{i,j}\; |\; i,j \in \{1,\ldots,n\}\}.
\end{align*}
Let $H_n$ be the graph obtained from $G_n$ by complementing the edges between
$B_2$ and $W_2$. By Fact~\ref{fact:bip}, the class of graphs $\{H_n\}_{n\geq 1}$ has
unbounded clique-width.
Note that $H_n[B_1 \cup W_2]$ and $H_n[B_2 \cup W_1]$ are $2P_2$-free bipartite graphs.
We claim that every $H_n$ is
$(P_6,\overline{2P_1+P_2})$-free.

First we show that $H_n$ is $(\overline{2P_1+P_2})$-free.
For contradiction, suppose that $\overline{2P_1+P_2}$ is present as an induced subgraph.
Consider one of the vertices of degree 3 in the $\overline{2P_1+P_2}$. It
cannot be in~$B_1$ or~$W_1$ since those vertices have neighbourhoods that are
independent sets. It cannot be a vertex in~$R_2$, since those vertices have
degree 2. Therefore one of these vertices must be in $B_2$ and the other in
$W_2$. Therefore the other two vertices in the diamond must both be in $R_2$,
which is a contradiction, since every vertex in $B_2$ has a unique neighbour
in~$R_2$. Therefore $H_n$ is indeed $(\overline{2P_1+P_2})$-free.

We now show that $H_n$ is $P_6$-free.
For contradiction, suppose that $P_6$ is present as an induced subgraph. We will first show that no
vertex of the $P_6$ may contain a vertex of $R_2$. Indeed, if one of the
vertices in the $P_6$ is in $R_2$, it must be an end-vertex of the $P_6$ (since
the neighbourhood of any vertex in $R_2$ induces a $P_2$, but $P_6$ does not
contain a $K_3$).  Let $x_1,\ldots,x_6$ be the vertices of the $P_6$, in order.
Note that $x_2,x_3,x_4,x_5 \not \in R_2$. Suppose that $x_1 \in R_2$. Without
loss of generality, we may assume $x_2 \in W_2$. If $x_3 \in B_1$, then we must
have $x_4 \in W_2$. But then there is no possible choice for $x_5$: we cannot
have $x_5 \in R_2$ (as noted above), we cannot have $x_5 \in B_2$ (since then
$x_2$ would be adjacent to~$x_5$) and we cannot have $x_5 \in B_1$, since then
$x_6$ would be in~$W_2$ and $H_n[x_2,x_3,x_5,x_6]$ would be a~$2P_2$,
contradicting the fact that $H_n[B_1 \cup W_2]$ is a 
$2P_2$-free bipartite graph.
Thus if $x_1 \in R_2$, $x_2 \in W_2$ then $x_3 \in B_2$ (since every vertex in
$W_2$ has a unique neighbour in $R_2$). Now $x_4 \not \in W_1$ (otherwise~$x_5$
would be in $B_2$, which would mean that $x_2$ would be adjacent to~$x_5$) and
$x_4 \not \in R_2$ (as explained above), so $x_4 \in W_2$. But this cannot
happen, since $x_5 \not \in R_2$ (as explained above), $x_5 \not \in B_2$
(since $x_5$ is not adjacent to $x_2$), so $x_5 \in B_1$, so $x_6 \in W_2$,
contradicting the fact that~$x_3$ and $x_6$ are not adjacent. We conclude that
no $P_6$ in $H_n$ can include a vertex of $R_2$.

By symmetry, any induced $P_6$ must therefore contain at least three vertices
in $W_1 \cup B_2$. In this case, it must have at least two vertices in $B_2$
since $W_1$ is an independent set. If the~$P_6$ also has a vertex in $W_2$ then it
must have exactly one vertex in $W_2$, two in $B_2$, none in $B_1$ and three in
$W_1$, which is impossible, by a parity argument. Thus the whole of the~$P_6$
must be contained in $H_n[W_1 \cup B_2]$, which leads to $H_n[W_1 \cup B_2]$
containing a $2P_2$, which contradicts the fact that $H_n[W_1 \cup B_2]$ is
a~$2P_2$-free bipartite graph. This completes the proof.\qed
\end{proof}

\begin{theorem}\label{thm:P_6P_2+P_43P_2gem}
The class of $(3P_2,P_2+P_4,P_6,\overline{P_1+P_4})$-free graphs has unbounded
clique-width.
\end{theorem}

\begin{proof}
Let $n\geq 1$ be an integer. 
Using the notation of Theorem~\ref{thm:generalunbounded}, we construct a
graph~$G_n$ as follows. We define vertex subsets
\begin{align*}
V_{0,0}&=\emptyset\\
V_{i,0}&=\{b_i\}\; \text{for}\; i\geq 1\\
V_{0,j}&=\{w_j\}\; \text{for}\; j\geq 1\\
V_{i,j}&=\{x_{i,j}\}\; \text{for}\; i,j\geq 1.
\end{align*}
We define edge subsets 
\begin{align*}
E_1&=\{b_ib_j\; |\; i,j \in \{1,\ldots,n\}, i \neq j\}\\
E_2&=\{w_iw_j\; |\; i,j \in \{1,\ldots,n\}, i \neq j\}\\
E_3&=\{b_kx_{i,j}\; |\; i,j,k \in \{1,\ldots,n\}, i\leq k\}\\
E_4&=\{w_kx_{i,j}\; |\; i,j,k \in \{1,\ldots,n\}, j\leq k\}.
\end{align*}
Let $V(G_n)$ be the union of the sets $V_{i,j}$ for $i,j \in \{0,\ldots,n\}$, and let
$E(G_n)=E_1\cup E_2 \cup E_3 \cup E_4$. By Theorem~\ref{thm:generalunbounded} with
$m=0$, the graph $G_n$ has clique-width at least $n$.

We define the sets
\begin{align*}
B&=\{b_i\; |\; i \in \{1,\ldots,n\}\}\\
W&=\{w_i\; |\; i \in \{1,\ldots,n\}\}\\
X&=\{x_{i,j}\; |\; i,j \in \{1,\ldots,n\}\}.
\end{align*}
Note that two vertices in $B$ (respectively $X$) cannot each have private
neighbours in $X$ (respectively~$B$). (When considering a pair of vertices
$v_1,v_2$, a {\em private neighbour} of $v_1$ is a vertex adjacent to $v_1$,
but not  to $v_2$.)
We will show that every~$G_n$ is $(3P_2,P_2+P_4,P_6,\overline{P_1+P_4})$-free.

First we show that $G_n$ is $(3P_2)$-free.
For contradiction, suppose that $G_n$ contains an induced $3P_2$.
Then, since $X$ is an independent set
and both $B$ and~$W$ are cliques, at most one of the $P_2$ components could
occur in each of $G_n[B \cup X]$ and $G_n[W \cup X]$. Since no vertex of $B$ is
adjacent to a vertex of~$W$, we find that~$G_n$ therefore cannot contain an induced $3P_2$.

We now show that $G_n$ is $(P_2+P_4)$-free.
For contradiction, suppose that $G_n$ contains an induced $P_2+P_4$. Since $X$ is an
independent set, we may assume that the $P_4$ contains at least one vertex
of~$B$.  The $P_4$ can have at most two vertices in~$B$ and if it has two such
vertices, one of them must be the end-vertex of the $P_4$; otherwise the two
vertices in $B$ would each have a private neighbour in $X$ 
which cannot happen.
Thus if the $P_4$ has a vertex in $B$ then it must have a vertex in~$X$ and
another in $W$ (since $X$ is an independent set). Thus the $P_4$ must have both
a vertex in $B$ and a vertex in $W$.
Then an independent $P_2$ cannot be found
since $B$ and $W$ are cliques and $X$ is an independent set.

We now show that $G_n$ is $P_6$-free.
For contradiction, suppose that $G_n$ contains an induced~$P_6$. Any~$P_6$ can contain at most
two vertices of $B$ (respectively~$W$), at most one of which can be adjacent to any
vertex of $X$ in the $P_6$. Let $v_1,\ldots,v_6$ be the vertices of the $P_6$
in order.  If the $P_6$ contains two vertices of~$B$ (respectively~$W$), then these
two vertices must be adjacent and one of them must be an end-vertex of the
$P_6$. In this case, assume without loss of generality that $v_1,v_2\in B$.
Then $v_3 \in X$, so $v_4 \in W$. Since $v_4$ is a middle-vertex of the~$P_6$,
neither $v_5,v_6 \not \in W$. This means $v_5,v_6 \in X$, which cannot happen
since $X$ is an independent set. This contradiction means that at most one
vertex of the~$P_6$ can be in each of $B$ and $W$, so at least four vertices of
the $P_6$ are members of~$X$. This is impossible since~$X$ is an independent
set. Thus $G_n$ is indeed $P_6$-free.

Finally, we show that $G_n$ is $(\overline{P_1+P_4})$-free.
For contradiction, suppose that~$G_n$ contains and induced $\overline{P_1+P_4}$. If the
dominating vertex of the $\overline{P_1+P_4}$ is in~$X$ then, since no vertex
in $B$ is adjacent to a vertex in $W$, the other vertices must either be all in
$B$ or all be in $W$, which is a contradiction. Thus the dominating vertex must
be (without loss of generality) in~$B$ and the other vertices in the
$\overline{P_1+P_4}$ must therefore all be in $B \cup X$. At most two of the
other vertices can be in~$X$ (since $X$ is an independent set and
$\overline{P_4}$ has independence number~2) and at most two of them can be in
$B$ (since $B$ is a clique). So exactly three vertices of the
$\overline{P_1+P_4}$ must be in $B$ and two must be in $X$. Since $X$ is an
independent set and $B$ is a clique, the two vertices in $X$ must be 
the two vertices of degree~2
in the $\overline{P_1+P_4}$. However, this means that each of these two
vertices in $X$ has a 
private neighbour in $B$, which is a contradiction. This
shows that $G_n$ is indeed $(\overline{P_1+P_4})$-free, which completes the
proof.\qed
\end{proof}

\section{Classifying Classes of $(H_1,H_2)$-Free Graphs}\label{sec:twographs}

In this section we study the boundedness of clique-width of classes of graphs defined by two
forbidden induced subgraphs. 
Recall that this study is partially motivated
by the fact that it is
easy to obtain a full classification for the boundedness of clique-width of
graph classes defined by one forbidden induced subgraph, as shown in the next
theorem.
This classification does not seem to have previously been explicitly stated in the literature.

\begin{theorem}\label{thm:single}
Let $H$ be a graph. The class of $H$-free graphs has bounded clique-width if and only if~$H$ is an induced subgraph of $P_4$.
\end{theorem}

\begin{proof}
First suppose that $H$ is an induced subgraph of $P_4$. Then the class of $H$-free graphs
is a subclass of the class of $P_4$-free graphs. The class of $P_4$-free
graphs is precisely the class of graphs of clique-width at most 2~\cite{CO00}.

Now suppose that $H$ is a graph such that the class of $H$-free graphs has bounded clique-width. By Fact~\ref{fact:comp}, the class of
$\overline{H}$-free graphs has bounded clique-width. By Lemma~\ref{l-classS}, $H,\overline{H} \in
\mathcal{S}$.
Since $\overline{H}\in \mathcal{S}$, the graph $\overline{H}$ must be
$(K_3,C_4)$-free. Thus~$H$ must be a $2P_2$-free forest whose maximum
independent set has size at most 2. Therefore $H$ must be one of the
following graphs: $P_1, 2P_1, P_1+P_2, P_2,P_3, P_4$. All these graphs are
induced subgraphs of $P_4$.\qed
\end{proof}

We are now ready to study classes of graphs defined by two forbidden induced
subgraphs.  Given four graphs $H_1,H_2,H_3,H_4$, we say that the class of
$(H_1,H_2)$-free graphs and the class of $(H_3,H_4)$-free graphs are {\em
equivalent} if the unordered pair $H_3,H_4$ can be obtained from the unordered
pair $H_1,H_2$ by some combination of the following operations:
\begin{enumerate}
\item complementing both graphs in the pair;
\item if one of the graphs in the pair is $K_3$,
replacing it with $\overline{P_1+P_3}$ or vice versa.
\end{enumerate} 
By Fact~\ref{fact:comp} and Lemma~\ref{l-paw}, if two classes are equivalent then one has
bounded clique-width if and only if the other one does.
Given this definition, we can now classify all classes defined by two forbidden
induced subgraphs for which it is known whether or not the clique-width is
bounded. This includes both the already-known results and 
our new results.
We will later show that (up to equivalence) this leaves only 13 open cases.

\begin{theorem}\label{thm:classification2}
Let ${\cal G}$ be a class of graphs defined by two forbidden induced subgraphs. Then:
\begin{enumeratei}
\item ${\cal G}$ has bounded clique-width if it is equivalent to a class of $(H_1,H_2)$-free graphs such that one of the following holds:
\begin{enumerate1}
\item \label{thm:classification2:bdd:P4} $H_1$ or $H_2 \ssi P_4$;
\item \label{thm:classification2:bdd:ramsey} $H_1=sP_1$ and $H_2=K_t$ for some $s,t$;
\item \label{thm:classification2:bdd:P_1+P_3} $H_1 \ssi P_1+P_3$ and $\overline{H_2} \ssi K_{1,3}+3P_1,\; K_{1,3}+P_2,\;\allowbreak P_1+\nobreak S_{1,1,2},\;\allowbreak P_6$ or $S_{1,1,3}$;
\item \label{thm:classification2:bdd:2P_1+P_2} $H_1 \ssi 2P_1+P_2$ and $\overline{H_2}\ssi 2P_1+P_3,\; 3P_1+P_2$ or $P_2+P_3$;
\item \label{thm:classification2:bdd:P_1+P_4} $H_1 \subseteq_i P_1+P_4$ and $\overline{H_2} \ssi P_1+P_4$ or $P_5$;
\item \label{thm:classification2:bdd:4P_1} $H_1 \subseteq_i 4P_1$ and $\overline{H_2} \ssi 2P_1+P_3$;
\item \label{thm:classification2:bdd:K_13} $H_1,\overline{H_2} \ssi K_{1,3}$.
\end{enumerate1}
\item ${\cal G}$ has unbounded clique-width if it is equivalent to a class of $(H_1,H_2)$-free graphs such that one of the following holds:
\begin{enumerate1}
\item \label{thm:classification2:unbdd:not-in-S} $H_1\not\in {\cal S}$ and $H_2 \not \in {\cal S}$;
\item \label{thm:classification2:unbdd:not-in-co-S} $\overline{H_1}\notin {\cal S}$ and $\overline{H_2} \not \in {\cal S}$;
\item \label{thm:classification2:unbdd:K_13or2P_2} $H_1 \si K_{1,3}$ or $2P_2$ and $\overline{H_2} \si 4P_1$ or $2P_2$;
\item \label{thm:classification2:unbdd:P_1+P_4} $H_1 \si P_1+P_4$ and $\overline{H_2} \si P_2+P_4$;
\item \label{thm:classification2:unbdd:2P_1+P_2} $H_1 \si 2P_1+P_2$ and $\overline{H_2} \si K_{1,3},\; 5P_1,\; P_2+P_4$ or $P_6$;
\item \label{thm:classification2:unbdd:3P_1} $H_1 \si 3P_1$ and $\overline{H_2} \si 2P_1+2P_2,\; 2P_1+P_4,\; 4P_1+P_2,\; 3P_2$ or $2P_3$;
\item \label{thm:classification2:unbdd:4P_1} $H_1 \si 4P_1$ and $\overline{H_2} \si P_1 + P_4$ or $3P_1+P_2$.
\end{enumerate1}
\end{enumeratei}
\end{theorem}
\begin{proof}
We first consider the bounded cases. 
Statement~\ref{thm:classification2:bdd:P4} follows from Theorem~\ref{thm:single}.
To prove Statement~\ref{thm:classification2:bdd:ramsey} note that if $H_1=sP_1$ and $H_2=K_t$ for some $s,t$ then by Ramsey's
Theorem, all graphs in the class of $(H_1,H_2)$-free graphs have a bounded number of vertices 
and therefore the clique-width of graphs in this class is bounded.
By the definition of equivalence, when proving Statement~\ref{thm:classification2:bdd:P_1+P_3}, we may assume that $H_1=K_3$. Then Statement~\ref{thm:classification2:bdd:P_1+P_3} follows from Fact~\ref{fact:comp} combined with the fact that $(K_3,H)$-free graphs have bounded
clique-width if~$H$ is
$K_{1,3}+3P_1$ (Theorem~\ref{thm:K_3s_111+3P_1}),
$K_{1,3}+P_2$~\cite{DLRR12},
$P_1+S_{1,1,2}$ (Theorem~\ref{thm:K_3s_112+P_1}),
$P_6$~\cite{BKM06} or
$S_{1,1,3}$~\cite{DLRR12}.
Statement~\ref{thm:classification2:bdd:2P_1+P_2} follows
from Fact~\ref{fact:comp} and the fact that $(\overline{2P_1+P_2},2P_1+P_3)$-free, $(\overline{2P_1+P_2},3P_1+P_2)$-free and $(\overline{2P_1+P_2},P_2+P_3)$-free graphs have bounded
clique-width~\cite{DHP0}.
Statement~\ref{thm:classification2:bdd:P_1+P_4} follows from Fact~\ref{fact:comp} and the fact that both
$(P_1+P_4,\overline{P_1+P_4})$-free graphs~\cite{BLM04b} and
$(P_5,\overline{P_1+P_4})$-free graphs~\cite{BLM04} have bounded clique-width.
Statement~\ref{thm:classification2:bdd:4P_1} follows from Fact~\ref{fact:comp} and the fact that $(2P_1+\nobreak P_3,\allowbreak K_4)$-free graphs have bounded clique-width~\cite{BDHP15}.
Statement~\ref{thm:classification2:bdd:K_13} follows from the fact that $(K_{1,3},\overline{K_{1,3}})$-free
graphs have bounded clique-width~\cite{BL02,BM02}.

We now consider the unbounded cases.
Statements~\ref{thm:classification2:unbdd:not-in-S} and~\ref{thm:classification2:unbdd:not-in-co-S} follow from Lemma~\ref{l-classS} and Fact~\ref{fact:comp}.
Statement~\ref{thm:classification2:unbdd:K_13or2P_2} follows from the fact that the classes of
$(C_4,\allowbreak K_{1,3},\allowbreak K_4,\allowbreak \overline{2P_1+P_2})$-free~\cite{BELL06},
$(K_4,2P_2)$-free~\cite{BELL06} and
$(C_4,C_5,2P_2)$-free graphs (or equivalently, split graphs)~\cite{MR99} have unbounded clique-width. 
Statement~\ref{thm:classification2:unbdd:P_1+P_4} follows from Fact~\ref{fact:comp} and the fact that the class of
$(P_2+\nobreak P_4,\allowbreak 3P_2,\allowbreak P_6,\allowbreak \overline{P_1+P_4})$-free (Theorem~\ref{thm:P_6P_2+P_43P_2gem})
graphs have unbounded clique-width.
Statement~\ref{thm:classification2:unbdd:2P_1+P_2} follows from Fact~\ref{fact:comp} and the fact that
$(C_4,K_{1,3},K_4,\overline{2P_1+P_2})$-free~\cite{BELL06},
$(5P_1,\overline{2P_1+P_2})$-free~\cite{DGP14},
$(\overline{2P_1+P_2},P_2+P_4)$-free (see arXiv version of~\cite{DHP0})
 and
$(P_6,\overline{2P_1+P_2})$-free (Theorem~\ref{thm:P_6diamond}) graphs have unbounded
clique-width.
To prove Statement~\ref{thm:classification2:unbdd:3P_1}, suppose $H_1 \si 3P_1$ and $\overline{H_2} \si
2P_1+2P_2, 2P_1+P_4, 4P_1+P_2,3P_2$ or~$2P_3$. Then $\overline{H_1} \not \in
{\cal S}$, so $\overline{H_2} \in {\cal S}$, otherwise we are done by
Statement~\ref{thm:classification2:unbdd:not-in-co-S}. By Lemma~\ref{l-addit}, 
$\overline{H_2}$ is not an induced subgraph of any graph in
$\{K_{1,3}+3P_1,K_{1,3}+P_2, P_1+S_{1,1,3}, S_{1,2,3}\}$.
The class of $(H_1,H_2)$-free graphs contains the class of complements of
$\overline{H_2}$-free bipartite graphs. By Fact~\ref{fact:comp} and 
Lemma~\ref{l-bipartite}, this latter class has unbounded clique-width.
Statement~\ref{thm:classification2:unbdd:4P_1} follows from the
Fact~\ref{fact:comp} and the fact that the classes of
$(K_4,P_1+P_4)$-free graphs (Lemma~\ref{l-k4-p1+p-4}) and
$(4P_1,\overline{3P_1+P_2})$-free graphs~\cite{DGP14}
have unbounded clique-width.\qed
\end{proof}

As we will prove in Theorem~\ref{t-opencases}, the above classification leaves exactly 13 open cases (up to equivalence).

\begin{oproblem}\label{oprob:twographs}
Does the class of $(H_1,H_2)$-free graphs have bounded clique-width when:
\begin{enumerate}
\item \label{oprob:twographs:3P_1} $H_1=3P_1, \overline{H_2} \in \{P_1+P_2+P_3,P_1+2P_2,P_1+P_5,P_1+S_{1,1,3},P_2+P_4,\allowbreak S_{1,2,2},\allowbreak S_{1,2,3}\}$;
\item \label{oprob:twographs:2P_1+P_2} $H_1=2P_1+P_2, \overline{H_2} \in \{P_1+P_2+P_3,P_1+2P_2,P_1+P_5\}$;
\item \label{oprob:twographs:P_1+P_4} $H_1=P_1+P_4, \overline{H_2} \in \{P_1+2P_2,P_2+P_3\}$ or
\item \label{oprob:twographs:2P_1+P_3} $H_1=\overline{H_2}=2P_1+P_3$.
\end{enumerate}
\end{oproblem}
Note that the two pairs $(3P_1,\overline{S_{1,1,2}})$ and $(3P_1,\overline{S_{1,2,3}})$, or equivalently, the two pairs $(K_3,S_{1,2,2})$ and $(K_3,S_{1,2,3})$
are the only pairs that correspond to open cases in which both $H_1$ and~$H_2$ are connected. 
We also observe the following. Let  $H_2 \in \{P_1+\nobreak P_2+\nobreak P_3,\allowbreak P_1+\nobreak 2P_2,\allowbreak P_1+\nobreak P_5,\allowbreak P_1+\nobreak S_{1,1,3},\allowbreak P_2+\nobreak P_4,\allowbreak S_{1,2,2},\allowbreak S_{1,2,3}\}$.
Lemma~\ref{l-bipartite} shows that all bipartite $H_2$-free graphs 
have bounded clique-width. Moreover, the graph $P_1+2P_2$ is an induced subgraph of $H_2$. Hence, for investigating 
whether the boundedness of the clique-width of bipartite $H_2$-free graphs can be extended to $(K_3,H_2)$-free graphs, the $H_2=P_1+2P_2$ case is 
the starting case.

\begin{theorem}\label{t-opencases}
Let ${\cal G}$ be a class of graphs defined by two forbidden induced
subgraphs. Then~${\cal G}$ is not equivalent to any of the classes listed in
Theorem~\ref{thm:classification2} if and only if it is equivalent to one of the
13 cases listed in Open Problem~\ref{oprob:twographs}.
\end{theorem}

\begin{proof}
It is easy to verify that none of the classes listed in Open
Problem~\ref{oprob:twographs} are equivalent to classes listed in
Theorem~\ref{thm:classification2}.

Let $H_1,H_2$ be graphs and let ${\cal G}$ be the class of $(H_1,H_2)$-free
graphs. Suppose~${\cal G}$ is not equivalent to any class listed in
Theorem~\ref{thm:classification2}.  Then $H_1 \in {\cal S}$ or $H_2 \in {\cal
S}$, otherwise Theorem~\ref{thm:classification2}.\ref{thm:classification2:unbdd:not-in-S} applies.  Similarly,
$\overline{H_1} \in {\cal S}$ or $\overline{H_2} \in {\cal S}$. If
$H_i,\overline{H_i} \in {\cal S}$ for some $i \in \{1,2\}$, then $H_i \ssi P_4$ (as shown in the proof of
Theorem~\ref{thm:single}), in which case Theorem~\ref{thm:classification2}.\ref{thm:classification2:bdd:P4}
applies.

Due to the definition of equivalence, for the remainder of the proof we may
assume without loss of generality that $H_1,\overline{H_2} \in {\cal S}$, but
neither is an induced subgraph of $P_4$. Furthermore,  we may assume that neither $H_1$ nor $\overline{H_2}$ is isomorphic
to $P_1+P_3$, as in this case the definition of equivalence would allow us to replace $P_1+P_3$ by $3P_1$. 
Also note that the situation for~$H_1$ and $\overline{H_2}$ is 
symmetric, i.e. if we exchanged these graphs, the resulting class would be
equivalent.

Suppose that $3P_1 \not \ssi H_1$. Then we must have that $H_1=2P_2$ (as $H_1\not\ssi P_4$). If
$\overline{H_2} \si K_{1,3}, 4P_1$ or $2P_2$ then
Theorem~\ref{thm:classification2}.\ref{thm:classification2:unbdd:K_13or2P_2} applies. Since $\overline{H_2} \in
{\cal S}$, we may therefore assume that $\overline{H_2}$ is a linear forest
which is $(4P_1,2P_2)$-free. This means that $\overline{H_2}$ is an induced
subgraph of $P_1+P_4$, in which case Theorem~\ref{thm:classification2}.\ref{thm:classification2:bdd:P_1+P_4}
applies (since $2P_2 \ssi P_5$).

We therefore assume that $3P_1 \ssi H_1,\overline{H_2}$. Now
$H_1,\overline{H_2}$ must be $(2P_1+\nobreak 2P_2,\allowbreak 2P_1+\nobreak P_4,\allowbreak 4P_1+\nobreak P_2,\allowbreak 3P_2,\allowbreak 2P_3)$-free,
otherwise Theorem~\ref{thm:classification2}.\ref{thm:classification2:unbdd:3P_1} would apply. Since
$H_1,\overline{H_2} \in {\cal S}$, by Lemma~\ref{l-addit}, each of
$H_1,\overline{H_2}$ must either contain no edges or be an induced subgraph of
(possibly different) graphs in $\{K_{1,3}+\nobreak 3P_1,\allowbreak K_{1,3}+\nobreak P_2,\allowbreak P_1+\nobreak S_{1,1,3},\allowbreak
S_{1,2,3}\}$. The induced subgraphs of graphs in $\{K_{1,3}+\nobreak 3P_1,\allowbreak K_{1,3}+\nobreak P_2,\allowbreak 
P_1+\nobreak S_{1,1,3},\allowbreak S_{1,2,3}\}$ are listed in Table~\ref{table:smallgraphs}.

\begin{table}
\begin{center}
%\begin{comment}
\begin{tabular}{VWVWVW}
Graph & Name &
Graph & Name &
Graph & Name \\\hline\hline

%7 vertices
\marginbox{5pt}{\scalebox{0.4}
{\begin{tikzpicture}[scale=1]
\renewcommand*{\EdgeLineWidth}{1.7pt}
\GraphInit[vstyle=Simple]
\SetVertexSimple[MinSize=6pt]
\Vertex[x=0,y=0]{a}
\Vertex[d=1,a=120]{b}
\Vertex[d=1,a=240]{c}
\Vertex[d=2,a=240]{c2}
\Vertex[d=1,a=0]{d}
\Vertex[d=2,a=0]{e}
\Vertex[d=3,a=0]{f}
\Edges(b,a,c,c2)
\Edges(a,d,e,f)
\end{tikzpicture}}}           & $S_{1,2,3}$      & 
\marginbox{5pt}{\scalebox{0.4}
{\begin{tikzpicture}[scale=1]
\renewcommand*{\EdgeLineWidth}{1.7pt}
\GraphInit[vstyle=Simple]
\SetVertexSimple[MinSize=6pt]
\Vertex[x=0,y=0]{a}
\Vertex[d=1,a=120]{b}
\Vertex[d=1,a=240]{c}
\Vertex[d=1,a=0]{d}
\Vertex[d=2,a=0]{e}
\Vertex[d=3,a=0]{f}
\Vertex[d=4,a=0]{g}
\Edges(b,a,c)
\Edges(a,d,e,f)
\end{tikzpicture}}}     & $P_1+S_{1,1,3}$  &
\marginbox{5pt}{\scalebox{0.4}
{\begin{tikzpicture}[scale=1]
\renewcommand*{\EdgeLineWidth}{1.7pt}
\GraphInit[vstyle=Simple]
\SetVertexSimple[MinSize=6pt]
\Vertices{circle}{a,b,c,d,e,f}
\Vertex[x=0,y=0]{z}
\Edges(a,z,c)
\Edge(e)(z)
\end{tikzpicture}}}     & $K_{1,3}+\nobreak 3P_1$   \\\hline

%6 vertices
\marginbox{5pt}{\scalebox{0.4}
{\begin{tikzpicture}[scale=1]
\renewcommand*{\EdgeLineWidth}{1.7pt}
\GraphInit[vstyle=Simple]
\SetVertexSimple[MinSize=6pt]
\Vertex[x=0,y=0]{a}
\Vertex[d=1,a=120]{b}
\Vertex[d=1,a=240]{c}
\Vertex[d=1,a=0]{d}
\Vertex[d=2,a=0]{e}
\Vertex[d=3,a=0]{f}
\Edges(b,a,c)
\Edges(a,d,e,f)
\end{tikzpicture}}}           & $S_{1,1,3}$     &
\marginbox{5pt}{\scalebox{0.4}
{\begin{tikzpicture}[scale=1]
\renewcommand*{\EdgeLineWidth}{1.7pt}
\GraphInit[vstyle=Simple]
\SetVertexSimple[MinSize=6pt]
\Vertex[x=0,y=0]{a}
\Vertex[d=1,a=120]{b}
\Vertex[d=1,a=240]{c}
\Vertex[d=1,a=0]{d}
\Vertex[d=2,a=0]{e}
\Vertex[d=3,a=0]{f}
\Edges(b,a,c)
\Edges(a,d,e)
\end{tikzpicture}}}     & $P_1+S_{1,1,2}$ &
\marginbox{5pt}{\scalebox{0.4}
{\begin{tikzpicture}[scale=1]
\renewcommand*{\EdgeLineWidth}{1.7pt}
\GraphInit[vstyle=Simple]
\SetVertexSimple[MinSize=6pt]
\Vertex[x=0,y=0]{a}
\Vertex[d=1,a=120]{b}
\Vertex[d=1,a=240]{c}
\Vertex[d=1,a=0]{d}
\Vertex[d=2,a=0]{e}
\Vertex[d=3,a=0]{f}
\Edges(b,a,c)
\Edges(a,d)
\Edges(e,f)
\end{tikzpicture}}}      & $K_{1,3}+P_2$   \\

\marginbox{5pt}{\scalebox{0.4}
{\begin{tikzpicture}[scale=1]
\renewcommand*{\EdgeLineWidth}{1.7pt}
\GraphInit[vstyle=Simple]
\SetVertexSimple[MinSize=6pt]
\Vertex[x=0,y=0]{a}
\Vertex[x=0,y=1]{b}
\Vertex[x=0,y=2]{c}
\Vertex[x=1,y=0]{d}
\Vertex[x=1,y=1]{e}
\Vertex[x=1,y=2]{f}
\Edges(a,b,c)
\Edges(b,e)
\Edges(a,d)
\Edges(c,f)
\end{tikzpicture}}}   & $S_{1,2,2}$     &
\marginbox{5pt}{\scalebox{0.4}
{\begin{tikzpicture}[scale=1]
\renewcommand*{\EdgeLineWidth}{1.7pt}
\GraphInit[vstyle=Simple]
\SetVertexSimple[MinSize=6pt]
\Vertices{circle}{a,b,c,d,e,f}
\Edges(d,c)
\Edges(d,a)
\Edges(d,e)
\end{tikzpicture}}}     & $K_{1,3}+\nobreak 2P_1$  &
\marginbox{5pt}{\scalebox{0.4}
{\begin{tikzpicture}[scale=1]
\renewcommand*{\EdgeLineWidth}{1.7pt}
\GraphInit[vstyle=Simple]
\SetVertexSimple[MinSize=6pt]
\Vertices{circle}{a,b,c,d,e,f}
\Edges(f,a,b,c,d,e)
\end{tikzpicture}}}             & $P_6$           \\

\marginbox{5pt}{\scalebox{0.4}
{\begin{tikzpicture}[scale=1]
\renewcommand*{\EdgeLineWidth}{1.7pt}
\GraphInit[vstyle=Simple]
\SetVertexSimple[MinSize=6pt]
\Vertices{circle}{a,b,c,d,e,f}
\Edges(b,c,d,e,f)
\end{tikzpicture}}}       & $P_1+P_5$       &
\marginbox{5pt}{\scalebox{0.4}
{\begin{tikzpicture}[scale=1,rotate=180]
\renewcommand*{\EdgeLineWidth}{1.7pt}
\GraphInit[vstyle=Simple]
\SetVertexSimple[MinSize=6pt]
\Vertices{circle}{a,b,c,d,e,f}
\Edges(b,c)
\Edges(d,e,f,a)
\end{tikzpicture}}}       & $P_2+P_4$       &
\marginbox{5pt}{\scalebox{0.4}
{\begin{tikzpicture}[scale=1]
\renewcommand*{\EdgeLineWidth}{1.7pt}
\GraphInit[vstyle=Simple]
\SetVertexSimple[MinSize=6pt]
\Vertices{circle}{a,b,c,d,e,f}
\Edges(a,b,c)
\Edges(e,f)
\end{tikzpicture}}} & $P_1+\nobreak P_2+\nobreak P_3$   \\

\marginbox{5pt}{\scalebox{0.4}
{\begin{tikzpicture}[scale=1]
\renewcommand*{\EdgeLineWidth}{1.7pt}
\GraphInit[vstyle=Simple]
\SetVertexSimple[MinSize=6pt]
\Vertices{circle}{a,b,c,d,e,f}
\Edges(c,d,e)
\end{tikzpicture}}}      & $3P_1+P_3$      &
\marginbox{5pt}{\scalebox{0.4}
{\begin{tikzpicture}[scale=1]
\renewcommand*{\EdgeLineWidth}{1.7pt}
\GraphInit[vstyle=Simple]
\SetVertexSimple[MinSize=6pt]
\Vertices{circle}{a,b,c,d,e,f}
\end{tikzpicture}}}            & $6P_1$\\\hline

%5 vertices
\marginbox{5pt}{\scalebox{0.4}
{\begin{tikzpicture}[scale=1]
\renewcommand*{\EdgeLineWidth}{1.7pt}
\GraphInit[vstyle=Simple]
\SetVertexSimple[MinSize=6pt]
\Vertex[x=0,y=0]{a}
\Vertex[d=1,a=120]{b}
\Vertex[d=1,a=240]{c}
\Vertex[d=1,a=0]{d}
\Vertex[d=2,a=0]{e}
\Edges(b,a,c)
\Edges(a,d,e)
\end{tikzpicture}}}           & $S_{1,1,2}$     &
\marginbox{5pt}{\scalebox{0.4}
{\begin{tikzpicture}[scale=1]
\renewcommand*{\EdgeLineWidth}{1.7pt}
\GraphInit[vstyle=Simple]
\SetVertexSimple[MinSize=6pt]
\Vertex[x=0,y=0]{a}
\Vertex[d=1,a=120]{b}
\Vertex[d=1,a=240]{c}
\Vertex[d=1,a=0]{d}
\Vertex[d=2,a=0]{e}
\Edges(b,a,c)
\Edges(a,d)
\end{tikzpicture}}}      & $K_{1,3}+P_1$   &
\marginbox{5pt}{\scalebox{0.4}
{\begin{tikzpicture}[scale=1,rotate=90]
\renewcommand*{\EdgeLineWidth}{1.7pt}
\GraphInit[vstyle=Simple]
\SetVertexSimple[MinSize=6pt]
\Vertices{circle}{a,b,c,d,e}
\Edges(d,e,a,b,c)
\end{tikzpicture}}}             & $P_5$           \\
\marginbox{5pt}{\scalebox{0.4}
{\begin{tikzpicture}[scale=1,rotate=90]
\renewcommand*{\EdgeLineWidth}{1.7pt}
\GraphInit[vstyle=Simple]
\SetVertexSimple[MinSize=6pt]
\Vertices{circle}{a,b,c,d,e}
\Edges(b,c,d,e)
\end{tikzpicture}}}       & $P_1+P_4$       &
\marginbox{5pt}{\scalebox{0.4}
{\begin{tikzpicture}[scale=1,rotate=90]
\renewcommand*{\EdgeLineWidth}{1.7pt}
\GraphInit[vstyle=Simple]
\SetVertexSimple[MinSize=6pt]
\Vertices{circle}{a,b,c,d,e}
\Edges(e,a,b)
\Edge(c)(d)
\end{tikzpicture}}}       & $P_2+P_3$       &
\marginbox{5pt}{\scalebox{0.4}
{\begin{tikzpicture}[scale=1,rotate=90]
\renewcommand*{\EdgeLineWidth}{1.7pt}
\GraphInit[vstyle=Simple]
\SetVertexSimple[MinSize=6pt]
\Vertices{circle}{a,b,c,d,e}
\Edges(e,a,b)
\end{tikzpicture}}}      & $2P_1+P_3$      \\
\marginbox{5pt}{\scalebox{0.4}
{\begin{tikzpicture}[scale=1,rotate=90]
\renewcommand*{\EdgeLineWidth}{1.7pt}
\GraphInit[vstyle=Simple]
\SetVertexSimple[MinSize=6pt]
\Vertices{circle}{a,b,c,d,e}
\Edge(b)(c)
\Edge(d)(e)
\end{tikzpicture}}}      & $P_1+2P_2$      &
\marginbox{5pt}{\scalebox{0.4}
{\begin{tikzpicture}[scale=1,rotate=90]
\renewcommand*{\EdgeLineWidth}{1.7pt}
\GraphInit[vstyle=Simple]
\SetVertexSimple[MinSize=6pt]
\Vertices{circle}{a,b,c,d,e}
\Edge(c)(d)
\end{tikzpicture}}}      & $3P_1+P_2$      &
\marginbox{5pt}{\scalebox{0.4}
{\begin{tikzpicture}[scale=1,rotate=90]
\renewcommand*{\EdgeLineWidth}{1.7pt}
\GraphInit[vstyle=Simple]
\SetVertexSimple[MinSize=6pt]
\Vertices{circle}{a,b,c,d,e}
\end{tikzpicture}}}            & $5P_1$          \\\hline

%4 vertices
\marginbox{5pt}{\scalebox{0.4}
{\begin{tikzpicture}[scale=1]
\renewcommand*{\EdgeLineWidth}{1.7pt}
\GraphInit[vstyle=Simple]
\SetVertexSimple[MinSize=6pt]
\Vertex[x=0,y=0]{a}
\Vertex[d=1,a=120]{b}
\Vertex[d=1,a=240]{c}
\Vertex[d=1,a=0]{d}
\Edges(b,a,c)
\Edges(a,d)
\end{tikzpicture}}}            & $K_{1,3}$       &
\marginbox{5pt}{\scalebox{0.4}
{\begin{tikzpicture}[scale=1,rotate=45]
\renewcommand*{\EdgeLineWidth}{1.7pt}
\GraphInit[vstyle=Simple]
\SetVertexSimple[MinSize=6pt]
\Vertices{circle}{a,b,c,d}
\Edges(d,a,b,c)
\end{tikzpicture}}}             & $P_4$           &
\marginbox{5pt}{\scalebox{0.4}
{\begin{tikzpicture}[scale=1,rotate=45]
\renewcommand*{\EdgeLineWidth}{1.7pt}
\GraphInit[vstyle=Simple]
\SetVertexSimple[MinSize=6pt]
\Vertices{circle}{a,b,c,d}
\Edges(d,a,b)
\end{tikzpicture}}}       & $P_1+P_3$       \\
\marginbox{5pt}{\scalebox{0.4}
{\begin{tikzpicture}[scale=1,rotate=45]
\renewcommand*{\EdgeLineWidth}{1.7pt}
\GraphInit[vstyle=Simple]
\SetVertexSimple[MinSize=6pt]
\Vertices{circle}{a,b,c,d}
\Edges(b,c)
\Edges(a,d)
\end{tikzpicture}}}            & $2P_2$          &
\marginbox{5pt}{\scalebox{0.4}
{\begin{tikzpicture}[scale=1,rotate=45]
\renewcommand*{\EdgeLineWidth}{1.7pt}
\GraphInit[vstyle=Simple]
\SetVertexSimple[MinSize=6pt]
\Vertices{circle}{a,b,c,d}
\Edges(a,b)
\end{tikzpicture}}}      & $2P_1+P_2$      &
\marginbox{5pt}{\scalebox{0.4}
{\begin{tikzpicture}[scale=1,rotate=45]
\renewcommand*{\EdgeLineWidth}{1.7pt}
\GraphInit[vstyle=Simple]
\SetVertexSimple[MinSize=6pt]
\Vertices{circle}{a,b,c,d}
\end{tikzpicture}}}            & $4P_1$          \\\hline

%3 vertices
\marginbox{5pt}{\scalebox{0.4}
{\begin{tikzpicture}[scale=1,rotate=90]
\renewcommand*{\EdgeLineWidth}{1.7pt}
\GraphInit[vstyle=Simple]
\SetVertexSimple[MinSize=6pt]
\Vertices{circle}{a,b,c}
\Edges(c,a,b)
\end{tikzpicture}}}             & $P_3$           &
\marginbox{5pt}{\scalebox{0.4}
{\begin{tikzpicture}[scale=1,rotate=90]
\renewcommand*{\EdgeLineWidth}{1.7pt}
\GraphInit[vstyle=Simple]
\SetVertexSimple[MinSize=6pt]
\Vertices{circle}{a,b,c}
\Edges(b,c)
\end{tikzpicture}}}       & $P_1+P_2$       &
\marginbox{5pt}{\scalebox{0.4}
{\begin{tikzpicture}[scale=1,rotate=90]
\renewcommand*{\EdgeLineWidth}{1.7pt}
\GraphInit[vstyle=Simple]
\SetVertexSimple[MinSize=6pt]
\Vertices{circle}{a,b,c}
\end{tikzpicture}}}            & $3P_1$          \\\hline

%2 vertices
\marginbox{5pt}{\scalebox{0.4}
{\begin{tikzpicture}[scale=1]
\renewcommand*{\EdgeLineWidth}{1.7pt}
\GraphInit[vstyle=Simple]
\SetVertexSimple[MinSize=6pt]
\Vertices{circle}{a,b}
\Edge(a)(b)
\end{tikzpicture}}}             & $P_2$           &
\marginbox{5pt}{\scalebox{0.4}
{\begin{tikzpicture}[scale=1]
\renewcommand*{\EdgeLineWidth}{1.7pt}
\GraphInit[vstyle=Simple]
\SetVertexSimple[MinSize=6pt]
\Vertices{circle}{a,b}
\end{tikzpicture}}}            & $2P_1$          &

%1 vertex
\marginbox{5pt}{\scalebox{0.4}
{\begin{tikzpicture}[scale=1]
\renewcommand*{\EdgeLineWidth}{1.7pt}
\GraphInit[vstyle=Simple]
\SetVertexSimple[MinSize=6pt]
\Vertex[x=0][y=0]{a}
\end{tikzpicture}}}             & $P_1$           \\
\end{tabular}
%\end{comment}
\end{center}
\caption{The induced subgraphs of $S_{1,2,3}, S_{1,1,3}+P_1$, $K_{1,3}+3P_1$ and $K_{1,3}+P_2$, arranged by number of vertices.}
\label{table:smallgraphs}
\end{table}

\begin{sloppypar}
First suppose that $H_1$ contains no edges.  Then $\overline{H_2}$  must
contain an edge, otherwise Theorem~\ref{thm:classification2}.\ref{thm:classification2:bdd:ramsey} would apply.
We first assume that $H_1=3P_1$.
If $\overline{H_2} \ssi K_{1,3}+\nobreak 3P_1,\allowbreak K_{1,3}+\nobreak P_2,\allowbreak P_1+\nobreak S_{1,1,2},\allowbreak P_6$ or
$S_{1,1,3}$, then Theorem~\ref{thm:classification2}.\ref{thm:classification2:bdd:P_1+P_3} applies.  This leaves
the cases where $\overline{H_2} \in
\{P_1+\nobreak P_2+\nobreak P_3,\allowbreak P_1+\nobreak 2P_2,\allowbreak P_1+\nobreak P_5,\allowbreak P_1+\nobreak S_{1,1,3},\allowbreak P_2+\nobreak P_4,\allowbreak S_{1,2,2},\allowbreak S_{1,2,3}\}$,
all of which are stated in Open Problem~\ref{oprob:twographs}.\ref{oprob:twographs:3P_1}. Now assume
$H_1=kP_1$ for $k\geq 4$. If $\overline{H_2} \si K_{1,3}, P_1+P_4, 3P_1+P_2$ or~$2P_2$,
Theorem~\ref{thm:classification2}.\ref{thm:classification2:unbdd:K_13or2P_2} or~\ref{thm:classification2}.\ref{thm:classification2:unbdd:4P_1}
applies.  Otherwise, $\overline{H_2}$ must be a $(P_1+P_4,3P_1+P_2,2P_2)$-free linear
forest, which (by assumption) is not an edgeless graph. As $\overline{H_2}\not\ssi P_4$ and $\overline{H_2} \neq P_1+P_3$, this means that
$\overline{H_2} \in \{2P_1+P_2,2P_1+P_3\}$. 
In both these cases, if $k=4$ then $\overline{H_2} \ssi 2P_1+P_3$, so
Theorem~\ref{thm:classification2}.\ref{thm:classification2:bdd:4P_1} applies;
if $k \geq 5$ then $2P_1+P_2 \ssi \overline{H_2}$, so
Theorem~\ref{thm:classification2}.\ref{thm:classification2:unbdd:2P_1+P_2}
applies.
\end{sloppypar}

By symmetry, we may therefore assume that neither $H_1$ nor $\overline{H_2}$
are edgeless. As stated above, in this case we may assume that both $H_1$ and
$\overline{H_2}$ are induced subgraphs of (possibly different) graphs in
$\{K_{1,3}+\nobreak 3P_1,\allowbreak K_{1,3}+\nobreak P_2,\allowbreak P_1+\nobreak S_{1,1,3},\allowbreak S_{1,2,3}\}$. Combining this with
our previous assumptions 
that neither $H_1$ nor $\overline{H_2}$ is equal to $P_1+P_3$ or an induced subgraph of $P_4$ 
means that $H_1,\overline{H_2} \in
\{K_{1,3},\allowbreak K_{1,3}+\nobreak P_1,\allowbreak K_{1,3}+\nobreak 2P_1,\allowbreak K_{1,3}+\nobreak 3P_1,\allowbreak K_{1,3}+\nobreak P_2,\allowbreak P_1+\nobreak P_2+\nobreak P_3,\allowbreak P_1+\nobreak 2P_2,\allowbreak P_1+\nobreak P_4,\allowbreak P_1+\nobreak P_5,\allowbreak P_1+\nobreak S_{1,1,2},\allowbreak P_1+\nobreak S_{1,1,3},\allowbreak 2P_1+\nobreak P_2,\allowbreak 2P_1+\nobreak P_3,\allowbreak 3P_1+\nobreak P_2,\allowbreak 3P_1+\nobreak P_3,\allowbreak P_2+\nobreak P_3,\allowbreak P_2+\nobreak P_4,\allowbreak P_5,\allowbreak P_6,\allowbreak S_{1,1,2},\allowbreak S_{1,1,3},\allowbreak S_{1,2,2},\allowbreak S_{1,2,3}\}$
(see also Table~\ref{table:smallgraphs} and recall that $3P_1\ssi\nobreak H_1,\overline{H_2}$).
In particular, this shows that the number of open cases is finite.

Suppose $H_1$ is not a linear forest. Then $K_{1,3}\ssi H_1$. If
$\overline{H_2} \si 2P_1+P_2, 4P_1$ or $2P_2$ then
Theorem~\ref{thm:classification2}.\ref{thm:classification2:unbdd:K_13or2P_2} or~\ref{thm:classification2}.\ref{thm:classification2:unbdd:2P_1+P_2}
applies.  The only remaining choice for $\overline{H_2}$ is $K_{1,3}$. 
Then, by symmetry, we may assume that $H_1$ is isomorphic to $K_{1,3}$, 
in which case Theorem~\ref{thm:classification2}.\ref{thm:classification2:bdd:K_13} applies.

We may now assume that $H_1$ and $\overline{H_2}$ are both linear forests, each
containing at least one edge. In other words, $H_1,\overline{H_2} \in
\{P_1+\nobreak P_2+\nobreak P_3,\allowbreak P_1+\nobreak 2P_2,\allowbreak P_1+\nobreak P_4,\allowbreak P_1+\nobreak P_5,\allowbreak 2P_1+\nobreak P_2,\allowbreak 2P_1+\nobreak P_3,\allowbreak 3P_1+\nobreak P_2,\allowbreak 3P_1+\nobreak P_3,\allowbreak P_2+\nobreak P_3,\allowbreak P_2+\nobreak P_4,\allowbreak P_5,\allowbreak P_6\}$.
Note that both of these graphs must therefore either be isomorphic to $P_5$ or
contain $2P_1+P_2$ as an induced subgraph. If $H_1=P_5$ then $\overline{H_2}$
must be $(4P_1,2P_2)$-free otherwise Theorem~\ref{thm:classification2}.\ref{thm:classification2:unbdd:K_13or2P_2}
applies. Thus $\overline{H_2} \in \{P_1+P_4,2P_1+P_2\}$, in which case
Theorem~\ref{thm:classification2}.\ref{thm:classification2:bdd:P_1+P_4} applies.  We may therefore assume that
neither $H_1$ nor $\overline{H_2}$ is isomorphic to $P_5$, and both must
therefore contain $2P_1+P_2$ as an induced subgraph.  Therefore, neither $H_1$
nor $\overline{H_2}$ may contain $5P_1, P_2+P_4$ or $P_6$ as an induced subgraph,
otherwise Theorem~\ref{thm:classification2}.\ref{thm:classification2:unbdd:2P_1+P_2} would apply.
We therefore conclude that $H_1,\overline{H_2} \in
\{P_1+\nobreak P_2+\nobreak P_3,\allowbreak P_1+\nobreak 2P_2,\allowbreak P_1+\nobreak P_4,\allowbreak P_1+\nobreak P_5,\allowbreak 2P_1+\nobreak P_2,\allowbreak 2P_1+\nobreak P_3,\allowbreak 3P_1+\nobreak P_2,\allowbreak P_2+\nobreak P_3\}$.

Suppose $H_1=2P_1+P_2$. If $\overline{H_2} \in \{P_1+P_4,2P_1+P_2\}$, then
Theorem~\ref{thm:classification2}.\ref{thm:classification2:bdd:P_1+P_4} would apply. If $\overline{H_2} \in \{2P_1+\nobreak P_3,\allowbreak 3P_1+\nobreak P_2,\allowbreak P_2+\nobreak P_3\}$, then Theorem~\ref{thm:classification2}.\ref{thm:classification2:bdd:2P_1+P_2} would apply. This leaves the cases where
$\overline{H_2} \in
\{P_1+\nobreak P_2+\nobreak P_3,\allowbreak P_1+\nobreak 2P_2,\allowbreak P_1+\nobreak P_5\}$, which
appear as Open Problem~\ref{oprob:twographs}.\ref{oprob:twographs:2P_1+P_2}. We now assume neither $H_1$ nor
$\overline{H_2}$ is isomorphic to $2P_1+P_2$.

Suppose $H_1=P_1+P_4$. If $\overline{H_2}\in
\{P_1+\nobreak P_2+\nobreak P_3,\allowbreak P_1+\nobreak P_5,\allowbreak 2P_1+\nobreak P_3,\allowbreak 3P_1+\nobreak P_2\}$
then Theorem~\ref{thm:classification2}.\ref{thm:classification2:unbdd:P_1+P_4} or~\ref{thm:classification2}.\ref{thm:classification2:unbdd:4P_1} would apply. If
$\overline{H_2}=P_1+P_4$, then Theorem~\ref{thm:classification2}.\ref{thm:classification2:bdd:P_1+P_4} applies. This
leaves the case where $\overline{H_2} \in
\{P_1+\nobreak 2P_2,\allowbreak P_2+\nobreak P_3\}$, both of which appear in Open
Problem~\ref{oprob:twographs}.\ref{oprob:twographs:P_1+P_4}. We may therefore assume that $H_1$ and
$\overline{H_2}$ are not isomorphic to $P_1+P_4$.

We have now that $H_1$ and $\overline{H_2} \in
\{P_1+\nobreak P_2+\nobreak P_3,\allowbreak P_1+\nobreak 2P_2,\allowbreak P_1+\nobreak P_5,\allowbreak 2P_1+\nobreak P_3,\allowbreak 3P_1+\nobreak P_2,\allowbreak P_2+\nobreak P_3\}$.  Note that
each of these graphs contains either $4P_1$ or $2P_2$ as an induced subgraph.
If either~$H_1$ or $\overline{H_2}$ contains an induced $2P_2$, then
in all these cases Theorem~\ref{thm:classification2}.\ref{thm:classification2:unbdd:K_13or2P_2} would apply. We may therefore
assume that $H_1,\overline{H_2} \in \{2P_1+P_3, 3P_1+P_2\}$. However, both
these graphs contain~$4P_1$, so if $H_1=3P_1+P_2$, then
Theorem~\ref{thm:classification2}.\ref{thm:classification2:unbdd:4P_1} applies. Therefore
$H_1=\overline{H_2}=2P_1+P_3$, which is Open Problem~\ref{oprob:twographs}.\ref{oprob:twographs:2P_1+P_3}. This completes the proof.\qed
\end{proof}

\section{Forbidding Other Patterns}\label{s-other}

Instead of forbidding one or more graphs as an induced subgraph of some other graph $G$, we could also forbid graphs under other containment relations. 
For example, a graph $G$ is
 {\it $(H_1,\ldots,H_p)$-subgraph-free} if $G$ has no subgraph isomorphic to a graph in $\{H_1,\ldots,H_p\}$.
In this section we consider this containment relation and two other well-known containment relations, which we define below.

Let $G$ and $H$ be graphs.
Then $G$ contains $H$ as a {\it minor} or {\it topological minor} if~$G$ can be modified into $H$ by 
a sequence that consists of edge contractions, edge deletions  and vertex deletions, or by  
a sequence that consists of vertex dissolutions, edge deletions and vertex deletions, respectively.
If $G$ does not contain any of the graphs $H_1,\ldots,H_p$ as a (topological) minor, we say that $G$ is
{\it $(H_1,\ldots,H_p)$-(topological-)minor-free}.

When we forbid a finite collection of either minors, subgraphs or topological minors, we can completely characterize those graph classes that have bounded clique-width. 
Before we prove these results we first state four known results, the last of which 
can be found in the textbook of Diestel~\cite{Di12}.
For a graph~$G$, let~$\tw(G)$ denote the tree-width of $G$ (see, for example, Diestel~\cite{Di12} for a definition). 

\newpage
\begin{lemma}[\cite{BL02}]\label{l-1g}
Let $H\in {\cal S}$. Then the class of $H$-subgraph-free graphs has bounded clique-width.
\end{lemma}

\begin{lemma}[\cite{CR05}]\label{l-treewidth}
Let $G$ be a graph. Then $\cw(G) \leq 3\times 2^{\tw(G)-1}$.
\end{lemma}

\begin{lemma}[\cite{RS86}]\label{l-minor}
Let $H$ be a planar graph. Then the class of $H$-minor-free graphs has bounded tree-width.
\end{lemma}

\begin{lemma}\label{l-diestel}
Let $H$ be a graph of maximum degree at most~$3$.
Then any graph that contains~$H$ as a minor contains $H$ as a topological minor.
\end{lemma}
 
We are now ready to state the three dichotomy results.

\begin{theorem}\label{t-finite}
Let $\{H_1,\ldots,H_p\}$ be a finite set of graphs. Then the following statements hold:
\begin{enumeratei}
\begin{sloppypar}
\item \label{t-finite-subgraph}The class of $(H_1,\ldots,H_p)$-subgraph-free graphs has bounded clique-width if and only if 
$H_i\in\nobreak {\cal S}$ for some $1\leq i\leq p$.\\[-8pt]
\end{sloppypar}
\item \label{t-finite-minor}The class of $(H_1,\ldots,H_p)$-minor-free graphs has bounded clique-width if and only if 
$H_i$ is planar for some $1\leq i\leq p$.\\[-8pt]
\item \label{t-finite-top minor}The class of $(H_1,\ldots,H_p)$-topological-minor-free graphs has bounded clique-width if and only if 
$H_i$ is planar and has maximum degree at most~$3$ for some $1\leq i\leq p$.
\end{enumeratei}
\end{theorem}

\begin{proof}
We first prove~\ref{t-finite-subgraph}. 
First suppose that $H_i\in {\cal S}$ for some $i$. Then the class of $(H_1,\ldots,H_p)$-subgraph-free graphs has bounded clique-width,
by Lemma~\ref{l-1g}.
Now suppose that $H_i\notin {\cal S}$ for all~$i$.
For $j \geq 0$, let $I_j$ be the graph formed from $2P_3$ by joining the central vertices of the two~$P_3$'s by a path of length $j$ (so $I_0=K_{1,4}$). Since $H_i \notin {\cal S}$, every $H_i$ contains an induced subgraph isomorphic to some $C_j$ or to some $I_j$.
Let $g$ be the maximum number of vertices of such an induced subgraph in $H_1+\cdots +H_p$.
Then the class of $(H_1,\ldots,H_p)$-subgraph-free graphs
contains the class of $g$-subdivided walls.
Hence, it has unbounded clique-width by Lemma~\ref{l-walls}.

We now prove~\ref{t-finite-minor}.
First suppose that $H_i$ is planar for some $i$.
Then the class of $H_i$-minor-free graphs, and thus the class of $(H_1,\ldots,H_p)$-minor-free graphs, has bounded tree-width by Lemma~\ref{l-minor}. 
Consequently, it has bounded clique-width, by Lemma~\ref{l-treewidth}.
Now suppose that $H_i$ is non-planar for all~$i$.
Because planar graphs are closed under taking minors, every planar graph is $(H_1,\ldots,H_p)$-minor-free.
Hence, the class of $(H_1,\ldots,H_p)$-minor-free graphs contains the class of walls, and thus has unbounded clique-width by Lemma~\ref{l-walls}.

Finally, we prove~\ref{t-finite-top minor}.
First suppose that $H_i$ is a planar graph of maximum degree at most~3 for some~$i$.
By Lemma~\ref{l-diestel}, any $H_i$-topological-minor-free is $H_i$-minor-free. Hence, we can repeat the arguments from above to find that the class of $(H_1,\ldots,H_p)$-free graphs has bounded clique-width.
Now suppose that $H_i$ is either non-planar or contains a vertex of degree at least 4 for all~$i$.
Consider some $H_i$. 
First assume that $H_i$ is not planar.
Because planar graphs are closed under taking topological minors, every planar graph, and thus every wall, is $H_i$-topological-minor-free.
Now suppose that~$H_i$ is planar. Then $H_i$ must have maximum degree at least~4. Because
every wall has minimum degree at most~3, it is $H_i$-topological-minor-free.
We conclude that the class of $(H_1,\ldots,H_p)$-topological-minor-free graphs contains the class of walls, and thus has unbounded clique-width by
Lemma~\ref{l-walls}.\qed
\end{proof}

\section{Consequences for Colouring}\label{s-conclusions}

One of the motivations of our research was to further the study of the computational complexity of the {\sc Colouring} problem for $(H_1,H_2)$-free 
graphs. Recall that {\sc Colouring} is polynomial-time solvable on any graph class of bounded clique-width by combining results of Kobler and Rotics~\cite{KR03b} 
and Oum~\cite{Oum08}.
By combining a number of known results~\cite{BGPS12a,BGPS12,DLRR12,GP14,KKTW01,Ma13,Ra04,RS02,Sc05}
with new results,
Dabrowski, Golovach and Paulusma~\cite{DGP14} presented a summary of known results for {\sc Colouring} restricted to $(H_1,H_2)$-free 
graphs.
Combining Theorem~\ref{thm:classification2} with the results of Kobler and Rotics~\cite{KR03b} and Oum~\cite{Oum08} and incorporating 
a number of recent results~\cite{HL,HJP14,Ma}
leads to an updated summary. This updated summary
(and a proof of it) 
can be found in the recent survey paper of Golovach, Johnson, Paulusma and Song~\cite{GJPS}, 
but for completeness we also present it here.

The graph $C_3^+$ is another notation used for $\overline{P_1+P_3}$.
The graph with vertices $a,b,c,d,e$ and edges $ab,ac,ad,bc,de$ is called the {\it hammer} and is denoted by~$C_3^*$.
The graph with vertices $a,b,c,d,e$ and edges $ab,ac,ad,bc,be$ is called the {\it bull} and is denoted by~$C_3^{++}$.

\begin{theorem}\label{t-twographs}
Let $H_1$ and $H_2$ be two graphs. Then the following hold:
\begin{enumeratei}
\item {\sc Colouring} is \NP-complete for $(H_1, H_2)$-free graphs if\\[-8pt]
\begin{enumerate1}
\item  $H_1\si C_r$ for $r\geq 3$ and $H_2\si C_s$ for  $s\geq 3$
\item $H_1\si K_{1,3}$ and $H_2\si K_{1,3}$
\item  $H_1$ and $H_2$ contain a spanning subgraph of $2P_2$ as an induced subgraph
\item $H_1\si C_3^{++}$ and $H_2\si K_{1,4}$
\item  $H_1\si C_3$ and $H_2\si K_{1,r}$  for  $r\geq 5$
\item $H_1\si C_r$ for  $r\geq 4$ and $H_2\si K_{1,3}$
\item  $H_1\si C_3$ and 
$H_2\si P_{22}$
\item $H_1\si C_r$ for  $r\geq 5$ and $H_2$ contains a spanning subgraph of $2P_2$ as an induced subgraph
\item $H_1\si C_r+ P_1$ for $3\leq r\leq 4$ or $H_1\si  \overline{C_r}$ for $r\geq 6$, and
$H_2$ contains a spanning subgraph of $2P_2$ as an induced subgraph
\item  $H_1\si K_4$ or $H_1\si \overline{2P_1+P_2}$, and $H_2\si K_{1,3}$.\\[-8pt]
\end{enumerate1}
\item {\sc Colouring} is polynomial-time solvable for $(H_1, H_2)$-free graphs if\\[-8pt] 
\begin{enumerate1}
\item $H_1$ or $H_2$ is an induced subgraph of $P_1+P_3$ or of $P_4$
\item $H_1\ssi K_{1,3}$, and $H_2\ssi C_3^{++}$, $C_3^*$ or $P_5$
\item $H_1\neq K_{1,5}$ is a forest on at most six vertices or $H_1= K_{1,3}+3P_1$, and $H_2\ssi C_3^+$  
\item $H_1\ssi sP_2$ or $H_1\ssi sP_1+P_5$ for  $s\geq 1$, and $H_2=K_t$ for $t\geq 4$
\item $H_1\ssi sP_2$ or $H_1\ssi sP_1+P_5$ for $s\geq 1$, and $H_2\ssi C_3^+$
\item $H_1\ssi P_1+P_4$ or $H_1\ssi P_5$, and $H_2\ssi  \overline{P_1+P_4}$ 
\item $H_1\ssi P_1+P_4$ or $H_1\ssi P_5$,
and $H_2\ssi \overline{P_5}$
\item $H_1\ssi 2P_1+P_2$, and $H_2\ssi \overline{3P_1+P_2}$ or $H_2\ssi \overline{2P_1+P_3}$
\item $H_1\ssi \overline{2P_1+P_2}$, and $H_2\ssi {3P_1+P_2}$ or $H_2\ssi 2P_1+P_3$
\item $H_1\ssi sP_1+P_2$ for $s\geq 0$ or $H_1=2P_2$, and  $H_2\ssi \overline{tP_1+P_2}$ for $t\geq 0$
\item $H_1\ssi 4P_1$ and $H_2 \ssi \overline{2P_1+P_3}$
\item $H_1\ssi P_5$, and $H_2\ssi C_4$ or $H_2\ssi \overline{2P_1+P_3}$.
\end{enumerate1}
\end{enumeratei}
\end{theorem}
From this summary we note that not only the case when $H_1=P_4$ or $H_2=P_4$ but thirteen
other maximal classes of $(H_1,H_2)$-free graphs for which {\sc Colouring} is known to be polynomial-time solvable can be obtained by combining 
Theorem~\ref{thm:classification2} with the results of Kobler and Rotics~\cite{KR03b} and Oum~\cite{Oum08} 
(see also~\cite{GJPS}).
One of these 
thirteen 
classes  is one that we obtained in this paper (Theorem~\ref{thm:K_3s_111+3P_1}), namely the class of  $(K_{1,3}+3P_1,\overline{P_1+P_3})$-free graphs, for which {\sc Colouring} was not previously known to  be polynomial-time solvable.
Note that  Dabrowski, Lozin, Raman and Ries~\cite{DLRR12} already showed that {\sc Colouring} is polynomial-time solvable for $(\overline{P_1+P_3},P_1+S_{1,1,2})$-free graphs, but in Theorem~\ref{thm:K_3s_112+P_1} we strengthened their result by showing that the clique-width of this class is also bounded.

Theorem~\ref{t-opencases} shows that there are~13 classes of $(H_1,H_2)$-free graphs (up to equivalence) 
for which we do not know whether their clique-width is bounded.
These classes correspond to 
28+6+4+1=39 distinct classes of $(H_1,H_2)$-free graphs. As can be readily verified from Theorem~\ref{t-twographs},  the complexity of {\sc Colouring} is unknown for only 
15 of these classes. We list these cases below:

\newpage
\begin{enumerate}
\item $\overline{H_1}\in \{3P_1,P_1+P_3\}$ and $H_2\in \{P_1+S_{1,1,3},S_{1,2,3}\}$;\\[-10pt]
\item $H_1=2P_1+\nobreak P_2$ and $\overline{H_2} \in \{P_1+\nobreak P_2+\nobreak P_3,\allowbreak P_1+\nobreak 2P_2,\allowbreak P_1+\nobreak P_5\}$;\\[-10pt]
\item $H_1=\overline{2P_1+\nobreak P_2}$ and $H_2 \in \{P_1+\nobreak P_2+\nobreak P_3,\allowbreak P_1+\nobreak 2P_2,\allowbreak P_1+\nobreak P_5\}$;\\[-10pt]
\item $H_1=P_1+P_4$ and $\overline{H_2} \in \{P_1+2P_2,P_2+P_3\}$;\\[-10pt]
\item $\overline{H_1}=P_1+P_4$ and ${H_2} \in \{P_1+2P_2,P_2+P_3\}$;\\[-10pt]
\item $H_1=\overline{H_2}=2P_1+P_3$.
\end{enumerate}
Note that Case 1 above reduces to two subcases by Lemma~\ref{l-olariu}. All classes of $(H_1,H_2)$-free graphs, for which the complexity of 
{\sc Colouring} is still open and which are not listed above have unbounded clique-width. Hence, new techniques will need to be developed to deal with these classes.

\section{Conclusions}\label{s-con}

We have determined for which pairs $(H_1,H_2)$ the class of $(H_1,H_2)$-free graphs has bounded clique-width, and for which pairs $(H_1,H_2)$ it has unbounded clique-width except for~13 non-equivalent cases, which we posed as open problems. We completely classified the (un)boundedness of the clique-width of those classes of graphs in which we forbid a finite family of graphs $\{H_1,\ldots,H_p\}$ as subgraphs, minors and topological minors, respectively.
Finally, we showed the implications of our results for the complexity of the {\sc Colouring} problem restricted to $(H_1,H_2)$-free graphs. In particular we
identified all 15 additional classes of $(H_1,H_2)$-free graphs for which {\sc Colouring} could potentially be solved in polynomial time if their clique-width turns out to be bounded.

\bibliographystyle{abbrv}

\bibliography{mybib}
\end{document}